\documentclass[10pt,journal,compsoc]{IEEEtran}

\usepackage[nocompress]{cite}                   
\usepackage{graphicx}                           
\graphicspath{ {Figures/} }

\usepackage{algorithm, algorithmic}             
\usepackage{amsthm,amsmath,amssymb,amsfonts}    
\usepackage{stfloats}                           

\newtheorem{assumption}{Assumption}             
\newtheorem{theorem}{Theorem}                   
\newtheorem{lemma}{Lemma}                       
\newtheoremstyle{noparens}                      
  {}{}%
  {\itshape}{}%
  {\bfseries}{.}%
  { }%
  {\thmname{#1}\thmnumber{ #2}\mdseries\thmnote{ #3}}
\theoremstyle{noparens}
\newtheorem{theoremNoParens}[theorem]{Theorem} 
\theoremstyle{remark}
\newtheorem{remark}{Remark}                     

\begin{document}

\title{Once and for All: Scheduling Multiple Users Using Statistical CSI under Fixed Wireless Access}


\author{Xin~Guan,
        Zhixing~Chen,
        Yibin~Kang,
        and~Qingjiang~Shi
}

\IEEEtitleabstractindextext{%
\begin{abstract}
Conventional multi-user scheduling schemes are designed based on instantaneous channel state information (CSI), indicating that decisions must be made every transmission time interval (TTI) which lasts at most several milliseconds. Only quite simple approaches can be exploited under this stringent time constraint, resulting in less than satisfactory scheduling performance. In this paper, we investigate the scheduling problem of a fixed wireless access (FWA) network using only statistical CSI. Thanks to their fixed positions, user terminals in FWA can easily provide reliable large-scale CSI lasting tens or even hundreds of TTIs. Inspired by this appealing fact, we propose an \emph{`once-and-for-all'} scheduling approach, i.e. given multiple TTIs sharing identical statistical CSI, only a single high-quality scheduling decision lasting across all TTIs shall be taken  rather than repeatedly making low-quality decisions every TTI. The proposed scheduling design is essentially a mixed-integer non-smooth non-convex stochastic problem with the objective of maximizing the weighted sum rate as well as the number of active users. We firstly replace the indicator functions in the considered problem by well-chosen sigmoid functions to tackle the non-smoothness. Via leveraging deterministic equivalent technique, we then convert the original stochastic problem into an approximated deterministic one, followed by linear relaxation of the integer constraints. However, the converted problem is still highly non-convex due to implicit equation constraints introduced by deterministic equivalent. To address this issue, we employ implicit optimization technique so that the gradient can be derived explicitly, with which we propose an algorithm design based on accelerated Frank-Wolfe method. Numerical results verify the effectiveness of our proposed scheduling scheme over state-of-the-art.
\end{abstract}

}

\maketitle

\IEEEdisplaynontitleabstractindextext

%
\IEEEpeerreviewmaketitle

\ifCLASSOPTIONcompsoc
\IEEEraisesectionheading{\section{Introduction}\label{sec:introduction}}
\else
\section{Introduction}
\label{sec:introduction}
\fi

%
%
%
%
\IEEEPARstart{M}{ultiple}-input multiple-output (MU-MIMO) network has found its wide applications in cellular networks in the last two decades \cite{spencer2004introduction,clerckx2013mimo,lim2013recent,castaneda2016overview,hu2020iterative}. One of the key aspects in efficiently enabling this prevailing multi-user transmission technique is effective downlink user scheduling \cite{castaneda2016overview}, by which a base station (BS) needs to dynamically select a proper subset of users from all candidate users for further resource allocation. Scheduling schemes in MU-MIMO networks attempt to allocate radio resources in a non-orthogonal fashion. To calculate key performance metrics of a scheduled user, such as bit-error-rate or transmission rate, full instantaneous channel state information (CSI) are required at the BS since those metrics heavily depend on CSI with respect to (w.r.t.) not only this user but also other co-scheduled users. In conventional mobile communication scenarios with high user mobility and rapidly changing wireless environment, instantaneous CSI will be easily outdated in a transmission time interval (TTI) lasting at most several milliseconds, indicating that user scheduling decisions must be repeatedly made every TTI. However, the stringent time constraint for CSI acquisition prohibits the exploitation of high-complexity approaches, while simple time-efficient scheduling schemes, such as greedy search, can only give 
suboptimal solutions with less than satisfactory scheduling performance.

\begin{figure}[htp]
    \centering
    \includegraphics[width=1.0\linewidth]{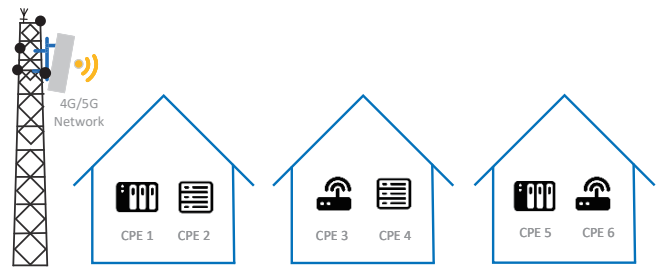}
    \caption{Illustration of Fixed Wireless Network (FWA), where the BS serves various CPEs at home with fixed-location. CPE: customer premise equipment. }
    \label{fig:sys_mdl_fwa}
\end{figure}

\emph{This bottleneck for conventional multi-user scheduling is expected to be alleviated under fixed wireless access (FWA) networks.} FWA, which uses 4G/5G technology to provide wireless last-mile Internet access, can be regarded as a cost-effective replacement for conventional fiber-to-the-home (FTTP) technology in less populated rural areas \cite{huawei20194G}. As illustrated in Fig.1, user terminals in FWA are fixed-position customer premise equipment (CPE) with relatively stationary wireless environment. This leads to an appealing fact that the statistical CSI for FWA users may stay quasi-static for tens or even hundreds of TTIs, although the instantaneous CSI still varies from one TTI to another. The much longer duration of statistical CSI gives us an opportunity to design more superior scheduling schemes with much higher complexity. Therefore, our concern is to develop efficient multi-user scheduling designs based on statistical CSI under the FWA scenario.

\subsection{Prior Work}
Existing research on feasible and efficient MU-MIMO user scheduling algorithm designs mainly focused on the case with instantaneous CSI, which can be roughly divided into two categories: \emph{greedy-based} and \emph{optimization-driven}.
Scheduling decisions in greedy-based designs is constructed by sequentially adding newly selected users when a locally maximum of the optimized utility is founded. The optimized utility can be defined in a direct approach (such as the weighted sum rate) or the indirect approach (such as various metric of spatial compatibility in \cite{castaneda2016overview} ). For the direct approach, one of the fundamental work belongs to the greedy user selection (GUS) based on sum-rate maximization proposed in \cite{dimic2005downlink} using zero-forcing (ZF) precoding and water-filling power allocation. In this scheduling design, the first selected user is the one with the highest rate. Subsequently, every newly selected user should maximize among all candidates the sum-rate rate of all already selected users, until stopping criteria is met. However, the algorithm in \cite{dimic2005downlink} might be computational demanding because calculation of the precoders and power allocation is required each iteration. Therefore, several works have modified the algorithm proposed in\cite{dimic2005downlink} to reduce complexity and improve performance, such as \cite{lau2005optimal,jiang2006greedy,hammarwall2007beamforming,wang2008user,tran2010multiuser,tran2012iterative,nam2014joint}. As for the indirect case, a common approach is the so-called semi-orthogonal user selection (SUS) originally proposed in \cite{yoo2006optimality}, which accomplishes user scheduling by iteratively pre-selecting a group of candidate users based on spatial clustering, where the candidate users at the next iteration will exhibit channel directions fulfilling an $\epsilon$-orthogonality criterion. Several variations of this approaches can also be found in the literature, such as \cite{sun2010eigen,driouch2011efficient,mao2012simplified}. Generally speaking, the greedy-based scheduling designs requires low computational complexity and is easy to implement. However, it does not guarantee neither performance nor convergence to the optimal solution. An alternative is to employ classical optimization theory. In fact, numerous works have included user scheduling for the optimization of weighted sum rate maximization \cite{chan2007capacity}, sum power minimization \cite{moretti2012efficient}, or some other metrics of spatial compatibility \cite{castaneda2014low}. However, all of the existing optimization-driven designs suffer from stringent running time limitations caused by instantaneous CSI, which prohibits their applications in practical networks. As a result, Developing efficient user scheduling algorithms relying on only statistical CSI is of crucial importance.

\subsection{Contributions}
In this paper, we develop a scheduling scheme for an MU-MIMO system with joint user pairing, layer determination and power control. In contrast to preliminary results with full instantaneous CSIT, we focus on a particular scenario where served users are quasi-static such that only stationary long-term statistical CSIT is available. The proposed scheduling design is conducted in a so-called \emph{once-and-for-all} manner, i.e. only a single high-quality scheduling decision
lasting across all TTIs sharing the same statistical CSI shall be taken rather than repeatedly
making low-quality decisions every TTI with various instantaneous CSI. The contributions of this paper are summarized as follows.

\begin{itemize}
    \item We propose an MU-MIMO scheduling design that simultaneously maximizes the weighted sum ergodic rate and the number of active served users by jointly optimizing  user pairing, layer determination and power control under statistical CSI. The original problem is highly non-tractable due to its randomness and non-smoothness. By replacing the ergodic rate term with its deterministic equivalent, this stochastic optimization problem is first approximated as a deterministic one. Moreover, the deterministic problem is further converted into its soft approximation where the non-smooth indicator functions in the problem are approximated by continuously differentiable sigmoid functions. Based on these approximations, we finally derive an deterministic and differentiable optimization problem but with integer optimization variables and implicit equality constraints.

    \item We perform linear relaxation to the integer variables and adopt implicit optimization techniques to deal with the aforesaid problem, followed by developing an efficient algorithm design based on accelerated Frank-Wolfe method. Our converted problem after deterministic equivalent and soft approximation is still  highly nonconvex due to integer variables and implicit constraints. We first perform linear relaxation of the integer variables so that they are converted into continuous variables. Then implicit optimization techniques is utilized such that the gradient for the problem can be computed explicitly. Subsequently we propose an algorithm design based on accelerated Frank-Wolfe method.
    
    \item  We test performance of our developed algorithms through extensive simulations with channels coefficients generated by Quadriga simulators. Our simulation results unveil that, compared with the conventional greedy-based benchmarks, the proposed user scheduling scheme outperforms the overall SE performance in the long term. 
\end{itemize}

\subsection{Notation}
Throughout the paper, we use the following notation. Lowercase and uppercase boldface letters represent column vectors and matrices, respectively. For a matrix $\mathbf{A}$, we denote its Hermitian and transpose by $\mathbf{A}^H$ and $\mathbf{A}^T$, respectively. We use $\mathbb{C}^{m \times n}$ to denote an $m$ by $n$ dimensional complex space. The notations $\mathbb{E} \{\cdot\}$ and $\text{Tr}(\cdot)$ represent expectation and trace operator.

\subsection{Paper Outline}
The rest of the paper is organized as follows. In Section II,
we introduce the system model and formulate the optimization problem. In Section III, we propose approximations of the original problem. In Section IV, we develop algorithms for the approximated problem. In Section V, we provide numerical simulations. We conclude in Section VI.

\section{System model And Problem formulation}
In this section, we describe the system model of the considered MU-MIMO system, including the signal model and the channel model, and then formulate our optimization problem.

\subsection{Signal Model}
Consider the downlink scheduling of an MU-MIMO system  with an $N_t$-antenna base station (BS) serving $I$ users, each equipped with $N_r$ antennas. The number of schedulable resource blocks groups (RBGs) is $N$. The number of transmitted data streams for each user on an RBG is not allowed to exceed a given threshold $L$. Denote $\alpha_{inl} \in \{0, 1\}$ as a binary variable indicating whether or not the $l$-th data stream for user $i$ transmitted from the BS is activated on RBG $n$, i.e.

\begin{equation}
    \alpha_{inl} = \left\{ 
    \begin{aligned}
        & 1, \ \text{if stream } l \text{ for user } i \text{ activated on RBG } n, \\
        & 0, \ \text{otherwise} .
    \end{aligned}
    \right.
\end{equation}
Without loss of generality, we assume that the precoding group size is the same as the RBG size. Specifically, let $\mathbf{g}_{inl} \in \mathbb{C}^{N_t \times 1}$ denote the precoding vector on RBG $n$ with respect to (w.r.t.) stream $l$ for user $i$ and  $s_{inl} \in \mathbb{C}$ denote the corresponding transmitted data symbol. It is assumed that $s_{inl}$ conforms to a zero-mean unit-variance circularly symmetric Gaussian distribution. Meanwhile, we have $\mathbb{E} \{ s_{inl} s_{jnq}^* \} = 0 $ for any $(j,q) \neq (i,l)$ with given $n$. Thus, the received signal of user $i$ on RBG $n$ can be expressed as
\begin{equation}
\begin{aligned}
    \mathbf{y}_{in} 
    & = \mathbf{H}_{in}  \alpha_{inl} \sqrt{p_{inl}} \mathbf{g}_{inl} s_{inl} \\
    & + \mathbf{H}_{in} \sum_{\left( j, q \right) \neq \left( i, l \right)} \alpha_{jnq} \sqrt{p_{jnq}} \mathbf{g}_{jnq} s_{jnq} + \mathbf{n}_{in}  
\end{aligned} \label{eq:sig_mdl}
\end{equation}
where $\mathbf{H}_{in} \in \mathbb{C}^{N_r \times N_t}$ is the random channel matrix between the BS and user $i$ on RBG $n$, $\mathbf{n}_{in} \in \mathbb{C}^{N_r \times 1} $ represents the corresponding white Gaussian noise vector with the covariance matrix given as $\sigma_{in}^2 \mathbf{I}$, $p_{inl}$ is the transmit power of stream $l$ for user $i$ on RBG $n$.

At the receiver side, the singular value decomposition (SVD) for the channel matrices are conducted such that $\mathbf{H}_{in} = \mathbf{U}_{in} \mathbf{\Lambda}_{in}^{1/2} \mathbf{V}_{in}^{H}$. The leftmost $L$ columns of $\mathbf{U}_{in}$, denoted as $\tilde{\mathbf{U}}_{in} \in \mathbb{C}^{N_r \times L} $, are exploited as the combiner for the received signal $\mathbf{y}_{in}$. Therefore, the estimated signal vector after combining can be written as
\begin{equation}
\begin{aligned}
    \tilde{\mathbf{y}}_{in} 
    & = \tilde{\mathbf{U}}_{in}^H \mathbf{y}_{in} \\
    & = \tilde{\mathbf{H}}_{in}  \alpha_{inl} \sqrt{p_{inl}} \mathbf{g}_{inl} s_{inl} \\
    & + \tilde{\mathbf{H}}_{in} \sum_{\left( j, q \right) \neq \left( i, l \right)} \alpha_{jnq} \sqrt{p_{jnq}} \mathbf{g}_{jnq} s_{jnq} + \tilde{\mathbf{n}}_{in} 
\end{aligned} \label{eq:sig_mdl_eff}
\end{equation}
where $\tilde{\mathbf{H}}_{in} = \tilde{\mathbf{U}}_{in}^H \mathbf{H}_{in} = \tilde{\mathbf{\Lambda}}_{in}^{1/2} \tilde{\mathbf{V}}_{in}^H$ with $\tilde{\mathbf{\Lambda}}_{in}^{1/2} 
\in \mathbb{C}^{L \times L} $ the diagonal matrix composed of the first $L$ diagonal elements of $\mathbf{\Lambda}_{in}^{1/2}$ and $\tilde{\mathbf{V}}_{in} \in \mathbb{C}^{N_t \times L} $ the matrix composed of the leftmost $L$ columns of $\mathbf{V}_{in}$, $\tilde{\mathbf{n}}_{in} = \tilde{\mathbf{U}}_{in}^H \mathbf{n}_{in} $ is the combined noise vector. Consequently, the $l$-th element of $\tilde{\mathbf{y}}_{in}$ is expressed as
\begin{equation}
\begin{aligned}
    \tilde{y}_{inl} 
    & = \tilde{\mathbf{h}}_{inl}^H  \alpha_{inl} \sqrt{p_{inl}} \mathbf{g}_{inl} s_{inl} \\
    & + \tilde{\mathbf{h}}_{inl}^H \sum_{\left( j, q \right) \neq \left( i, l \right)} \alpha_{jnq} \sqrt{p_{jnq}} \mathbf{g}_{jnq} s_{jnq} + \tilde{n}_{inl} 
\end{aligned} \label{eq:sig_mdl_eff_sigl_str}
\end{equation}
where $\tilde{\mathbf{h}}_{inl}^H$ is the $l$-th row of $\tilde{\mathbf{H}}_{in}$. 

It is worth noting that the conventional zero-forcing (ZF) \cite{spencer2004zero} scheme cannot be directly applied to get the multi-user precoders $\mathbf{g}_{inl}$ in \eqref{eq:sig_mdl_eff_sigl_str}. The underlying reason is that the coupling of the binary variables $\alpha_{inl}$ and $\tilde{\mathbf{h}}_{inl}^H$ in \eqref{eq:sig_mdl_eff_sigl_str} makes it unattainable to implement the pseudo inverse step of ZF precoding. As a countermeasure, we exploit the regularized zero-forcing (RZF) with a very small regularization factor to act as an approximation of the ZF scheme. Specifically, let $\delta > 0$ denote the regularization factor, the precoding vector $\mathbf{g}_{inl}$ should thus take the following form:
\begin{equation}
    \mathbf{g}_{inl} =\xi_{inl} {\mathbf{W}}_{n}^{-1}  \tilde{\mathbf{h}}_{inl} \label{eq:rzf}
\end{equation}
where ${\mathbf{W}}_{n} \in \mathbb{C}^{N_t \times N_t}$ is defined as 
\begin{equation}
    {\mathbf{W}}_{n} \triangleq \sum_{j=1}^{I} \sum_{q=1}^{L} \alpha_{jnq} \tilde{\mathbf{h}}_{jnq} \tilde{\mathbf{h}}_{jnq}^{H}+\delta \mathbf{I}
\end{equation}
and $\xi_{inl}$ is the power normalization scalar satisfying
\begin{equation}
    \xi_{inl}^2 = \frac{1}{  \tilde{\mathbf{h}}_{inl}^H {\mathbf{W}}_{n}^{-2} \tilde{\mathbf{h}}_{inl}} . \label{eq:rzf_coff}
\end{equation}
Combining equations \eqref{eq:sig_mdl_eff_sigl_str} and \eqref{eq:rzf}, we then derive the signal-to-interference-plus-noise ratio (SINR) as
\begin{equation}
\gamma_{inl}=\frac{ \alpha_{inl} p_{inl} \xi_{inl}^2 \left|\tilde{\mathbf{h}}_{inl}^{H} {\mathbf{W}}_{n}^{-1} \tilde{\mathbf{h}}_{inl} \right|^{2}}{\sum\limits_{(j, q) \neq(i, l)} \alpha_{jnq} p_{jnq} \xi_{jnq}^2 \left|\tilde{\mathbf{h}}_{inl}^{H} {\mathbf{W}}_{n}^{-1} \tilde{\mathbf{h}}_{jnq} \right|^{2}+\sigma_{i n }^{2}} .
\label{eq:sinr_rzf}
\end{equation}
Subsequently, the ergodic rate of user $i$ is defined as
\begin{equation}
    R_{i} = \mathbb{E}_{\tilde{\mathbf{h}}_{inl}} \left\{ \sum_{n,l} \log \left( 1+\gamma_{inl} \right) \right\}. \label{eq:rate}
\end{equation}
 
\subsection{Channel Model}

We make the following two assumptions about the random channel model used in this paper:
\begin{itemize}
    \item[-] Although the instantaneous channel coefficients randomly change with time, the corresponding channel correlation matrices are assumed to be slowly varying and thus are supposed to be constant and perfectly known at the transmitter during the channel coherence time. 
    \item[-] For a given stream $l$ w.r.t. a given user $i$, the corresponding channel correlation matrix is identical at all RBGs, which is denoted by ${\mathbf{R}}_{il}$.
\end{itemize}
 
Based on the above ssumptions, we exploit the following channel model throughout our work (Note that the similar model has been widely adopted in the literature \cite{couillet2011random}):

\begin{equation}
    \tilde{\mathbf{h}}_{inl} = \sqrt{N_t} {\mathbf{R}}_{il}^{1/2} \mathbf{x}_{inl} \label{eq:chanl_mdl}
\end{equation}
where $\mathbf{x}_{inl}$ conform to a circularly
symmetric complex Gaussian distribution with zero mean and covariance matrix given as $(1/Nt) \cdot \mathbf{I}$. 

\begin{remark}
In a practical communication system, the channel correlation matrix ${\mathbf{R}}_{il}$ can be well approximated by the corresponding sampling correlation matrix computed from observed channel vector samples. Specifically, suppose we have $T$ samples of channel measurements $\{\mathbf{H}_{in}^{(t)}\}_{t=1}^T$, where $T$ is large enough. The channel vector w.r.t. stream $l$ for user $i$ of sample $t$ can be equivalently written as 
\begin{equation}
    \left(\tilde{\mathbf{h}}_{inl}^{(t)}\right)^H = \left(\lambda_{inl}^{(t)}\right)^{1/2}\left(\mathbf{v}_{inl}^{(t)}\right)^H
\end{equation}
where $(\tilde{\mathbf{h}}_{inl}^{(t)})^H$ and $(\mathbf{v}_{inl}^{t})^H$ are the $l$-th singular value and right singular vector of the original channel matrix $\mathbf{H}_{in}^{(t)}$, respectively. The channel correlation matrix ${\mathbf{R}}_{il}$ can then be approximately computed by 
\begin{equation}
    \begin{aligned}
        {\mathbf{R}}_{il} & \approx \frac{1}{N  \cdot T} \sum_{n,t} \tilde{\mathbf{h}}_{inl}^{(t)} \left( \tilde{\mathbf{h}}_{inl}^{(t)}\right)^H \\
        & \approx \frac{1}{N  \cdot T} \sum_{n,t} \lambda_{inl}^{(t)} \mathbf{v}_{inl}^{(t)} \left( \mathbf{v}_{inl}^{(t)}\right)^H .
    \end{aligned}
\end{equation}
\end{remark}

\subsection{Problem Formulation}

Our main objective in this paper is to maximize the weighted sum rate of the MU-MIMO system, while ensuring that the number of active users being as large as possible. To this end, our optimization problem is formulated as
\begin{subequations}
\begin{align}
   \mathop{\text{max}}_{\{\alpha_{inl}, p_{inl}\}} \ & {\sum_{i=1}^I  \mu_i R_i+w\sum_{i=1}^I \mathbb{I}\left(\sum_{n,l}\alpha_{inl}\right)} \label{obj:orig} 
  \\
 \text{s. t.} \quad \ \ & \sum_{i, l} \alpha_{inl} \leq L_{\text{max}}, \quad \forall \ n,  \label{constr:str}
  \\
  & \sum_{i,n,l} p_{inl} \leq P_{\text{max}},\label{constr:pow} 
  \\
   & \alpha_{i n l} \in\{0,1\}, \quad \forall \ i,n,l, \label{constr:int}
  \\
  & p_{inl} \geq 0 , \quad \forall \ i,n,l, \label{constr:powge0}
  \\
& {R_i \geq \mathbb{I}\left(\sum_{n,l}\alpha_{inl}\right)R^{\text{min}}_{i}, \quad \forall i},
\label{constr:rate}
 \end{align}  \label{prblm:orig}
\end{subequations}
where 
\begin{equation}
    \mathbb{I}(x)=
\left\{
             \begin{array}{lr}
             1,\quad \text{if} \ x\neq0 ,\\
             0,\quad \text{if} \ x = 0
             
             \end{array}
\right.
\end{equation}
is the indicator function, $w$ is the penalty coefficient, $\mu_i$ is the weighting coefficient adjusting user $i$'s priority. Constraints \eqref{constr:str} indicate that the total number of streams at each RBG is not allowed to exceed the given threshold $L_{\text{max}}$. Constraint \eqref{constr:pow} restricts the total transmit power; Constraints \eqref{constr:int} and \eqref{constr:powge0} are the integer constraints for $\{\alpha_{inl}\}$ and the non-negative constraints for $\{p_{inl}\}$, respectively. Constraints \eqref{constr:rate} implies that once a user is activated, its corresponding data rate should be no less than the minimum rate requirement $R^{\text{min}}_{i}$.

\section{Problem Approximation}

Problem \eqref{prblm:orig} is highly non-tractable mainly because: 1) the non-smoothness and non-convexity caused by the indicator functions in the objective function \eqref{obj:orig} and the constraints \eqref{constr:rate}; 2) the lack of closed form for the rate expression $R_i$ due to the randomness of the channel coefficient $\tilde{\mathbf{h}}_{inl}$. In this section, we propose efficient approximations for the original problem \eqref{prblm:orig} to derive a more solvable approximated optimization problem.

\subsection{Soft Approximation Using Sigmoid Functions}
In this section, we focus on proposing a soft approximation version of the original problem \eqref{prblm:orig} to attack the non-smoothness caused by the indicator functions. First, constraints \eqref{constr:rate} indicates that a user is activated if and only if its rate exceeds the minimum rate requirement, i.e. 
\begin{equation}
    R_i - R_i^{\text{min}} \geq 0, \quad \forall i \ \text{with}  \ \sum_{n,l}\alpha_{inl} > 0. 
\end{equation}
Moreover, recall that $\mathbb{I}(x)$ can be well approximated by the sigmoid function $1/(1+e^{-\theta x})$ where $\theta$ is the predefined hyper-parameter, as illustrated in Fig.  \ref{fig:indi_fun_vs_sigmoid_fun}. Based on these observations, we replace the second term in \eqref{obj:orig} with its soft approximation based on the sigmoid function and formulate the following optimization problem:
\begin{subequations}
\begin{align}
   \mathop{\text{max}}_{\{\alpha_{inl}, p_{inl}\}} \ & \sum_{i=1}^I  \mu_i {R}_i+w\sum_{i=1}^I \frac{1}{1+e^{-\theta_i ({R}_i -R_i^{\min})}}  
  \\
 \text{s. t.} \quad \ \ & \eqref{constr:str} - \eqref{constr:powge0}.
 \end{align} \label{prblm:soft_appro}
\end{subequations}

\begin{figure}[htp]
    \centering
    \includegraphics[width=1.0\linewidth]{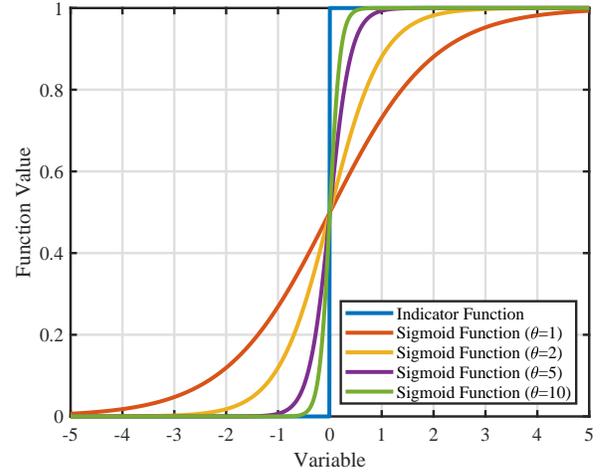}
    \caption{Comparison among the indicator function and sigmoid functions with varying hyper-parameter $\theta$ .  }
    \label{fig:indi_fun_vs_sigmoid_fun}
\end{figure}

Comparing problem \eqref{prblm:soft_appro} with the original problem \eqref{prblm:orig}, we observe that both problems imply the same physical meaning: If a user is not activated, its corresponding regularization term is (nearly) equal to zero; Otherwise, it is (nearly) equal to one and the corresponding rate exceeds the minimum rate requirement. In other words, problem \eqref{prblm:soft_appro} is a nice smooth version of the original non-tractable optimization problem \eqref{prblm:orig}. Therefore, in the following contexts we turn to solve the approximated problem \eqref{prblm:soft_appro}.

\subsection{Deterministic Equivalent of the SINR Expression}

Besides the non-smoothness due to the indicator functions, problem \eqref{prblm:orig} is still highly non-tractable due to the lack of closed form for $R_i$ as shown in \eqref{eq:rate}. This is because the randomness of the SINR expression $\gamma_{inl}$, or more specifically, the randomness of $\tilde{\mathbf{h}}_{inl}$. The key idea to address this issue is to replace the random term $\gamma_{inl}$ by its \emph{deterministic equivalent} $\bar{\gamma}_{inl}$, which is considered to be a deterministic expression, such that $\gamma_{inl} - \bar{\gamma}_{inl} \mathop{\longrightarrow}\limits_{N_t \rightarrow \infty} 0$\footnote{Although this asymptotic expression is derived in the limit regime, it usually shows nice approximation accuracy in finite size cases \cite{couillet2011random}. }. To accomplish this, we first review the highly useful theorem as detailed below. 

\begin{theoremNoParens} [{\cite[Th. 1]{wagner2012large}}]
   Let $\mathbf{S} \in \mathbb{C}^{K \times K}$ and $\mathbf{R}_d \in \mathbb{C}^{K \times K}$ be deterministic Hermitian non-negative definite matrices. Define  $\mathbf{B} \triangleq \sum_{d=1}^D \tilde{\mathbf{h}}_d \tilde{\mathbf{h}}_d^H  + \mathbf{S}$, where $\tilde{\mathbf{h}}_d^H = \mathbf{x}_d^H {\mathbf{R}}^{1/2}_d $ with $\mathbf{x}_d \in \mathbb{C}^{K}$ conforming to complex Gaussian distribution with zero mean and covariance matrix $(1/K) \cdot \mathbf{I}$.  Assume that the matrices $\mathbf{Q} \in \mathbb{C}^{K \times K}$ and $\mathbf{R}_d$ have uniformly bounded spectral norm. Define 
 \begin{equation}
    m(z)\triangleq\frac{1}{K}\mathrm{Tr}\left(
    \mathbf{Q}(\mathbf{B}-z\mathbf{I})^{-1}\right).
\end{equation}
Then, for any $z \in \mathbb{C} \setminus \mathbb{R}^{+}$, as $D$ and $K$ grow large with the constant ratio $c \triangleq K / D$ such that $0 < \lim_{K \to \infty} c < \infty$, we have that 
\begin{equation}
    m(z)-\bar{m}(z)
    \mathop{\longrightarrow}\limits_{K\rightarrow\infty} 0, \quad \bar{m}(z)=\frac{1}{K}\mathrm{Tr} \left( \mathbf{Q} \mathbf{T} \right)
\end{equation} 
almost surely, where
\begin{equation}
    \mathbf{T} =  \left(
    \frac{1}{K}\sum^D_{s=1}\frac{\mathbf{R}_s}{1+{e}_{s}(z)}
    +\mathbf{S}-z\mathbf{I}\right)^{-1} 
\end{equation}%
with the functions ${e}_{1}(z),{e}_{2}(z),...{e}_{D}(z)$ defined as the unique solution of the following equations:
\begin{equation}
   {e}_{d}(z)=\frac{1}{K}\mathrm{Tr}
    \left(\mathbf{R}_{d} \left(
    \frac{1}{K}\sum^D_{s=1}\frac{\mathbf{R}_s}{1+{e}_{s}(z)}
    +\mathbf{S}-z\mathbf{I}\right)^{-1} \right) \label{eq:stieltjes_transform}
\end{equation}%
with initial values ${e}_{d}^{(0)}(z) = -1/z$ for all $d$. \label{theorem:de}
\end{theoremNoParens}

In order to exploit the aforementioned theorem, several assumptions about the channel vectors and channel correlation matrices are further need (Similar assumptions have been widely seen in the existing literature \cite{wagner2012large}):

\begin{assumption}
All correlation matrices $\mathbf{R}_{il}$ have uniformly bounded spectral norm on $N_t$ and $I$, i.e. \label{assumption:norm1}
\begin{equation}
    \mathop{\lim\sup}_{N_t,I \rightarrow \infty} \sup\limits_{ \substack{1 \leq i \leq I \\ 1 \leq l \leq L}} \| \mathbf{R}_{il} \| < \infty .
\end{equation}
\end{assumption}
\begin{assumption}
The random channel vectors $\tilde{\mathbf{h}}_{inl}$ satisfy
\begin{equation}
    \mathop{\lim\sup}_{N_t,I \rightarrow \infty} \left \| \frac{1}{N_t} \sum_{i,l} \tilde{\mathbf{h}}_{inl} \tilde{\mathbf{h}}_{inl}^H \right \| < \infty
\end{equation}
with probability one. \label{assumption:norm2}
\end{assumption}

With the help of the two aforesaid assumptions and Theorem \ref{theorem:de}, the derivation for the deterministic equivalent of $\gamma_{inl}$ can be summarized as the following theorem.

\begin{theorem}
    Let assumptions \ref{assumption:norm1} and \ref{assumption:norm2} hold true. Then $\gamma_{inl} - \bar{\gamma}_{inl} \mathop{\longrightarrow}\limits_{N_t \rightarrow 0} 0$ almost surely with $\bar{\gamma}_{inl}$ defined as 
    \begin{equation}
        \bar{\gamma}_{inl} = \frac{ N_{t}\alpha_{inl} p_{inl}
    {e}_{inl}^2(1+{e}_{inl})^2}{{e}'_{inl} b_{inl}
    +{e}'_{inl}\sigma_{in}^{2}(1+{e}_{inl})^2}, \label{eq:de_sinr}
    \end{equation}
    where
\begin{align}
    e_{inl} & = \frac{1}{N_t} \mathrm{Tr}\left({\mathbf{R}}_{il}\mathbf{T}_{n} \right), \forall i,n,l, \label{eq:e} \\
    e'_{inl} & = \frac{1}{N_t} \mathrm{Tr}\left({\mathbf{R}}_{il}\mathbf{T}'_{n} \right), \forall i,n,l,  \label{eq:e'} \\
    e'_{jnq,inl} & = \frac{1}{N_t} \mathrm{Tr}\left({\mathbf{R}}_{il}\mathbf{T}'_{inl} \right), \forall i,n,l, \label{eq:e_il'} \\
    b_{inl} & = \sum\limits_{(j, q) \neq(i, l)} \frac{ \alpha_{jnq}p_{jnq}{e}'_{jnq,inl}}{{e}'_{jnq}}, \forall i,n,l.
\end{align}
with 
\begin{align}
\mathbf{T}_{n} & =\left( \frac{1}{N_t}\sum_{j=1}^I\sum_{q=1}^L\frac{\alpha_{jnq}{\mathbf{R}}_{jq}}{1+e_{jnq}}+\delta \mathbf{I}  \right)^{-1} \label{eq:T} \\
\mathbf{T}'_{n} & =\mathbf{T}_{n}\left(\frac{1}{N_t}\sum_{j=1}^I\sum_{q=1}^L\frac{\alpha_{jnq}e'_{jnq}{\mathbf{R}}_{jq}}{(1+e_{jnq})^2}+\mathbf{I} \right)\mathbf{T}_{n}  \label{eq:T'} \\
\mathbf{T}'_{inl} & =\mathbf{T}_{n}\left(\frac{1}{N_t}\sum_{j=1}^I\sum_{q=1}^L\frac{\alpha_{jnq}e'_{jnq,inl}{\mathbf{R}}_{jq}}{(1+e_{jnq})^2}+\mathbf{I} \right)\mathbf{T}_{n} . \label{eq:T_il'} 
\end{align} \label{theorem:DE_SINR}
\end{theorem} 

\begin{IEEEproof}
The proof is given in appendix \ref{app:theorem_de_sinr}.
\end{IEEEproof}

Though being a deterministic expression, computing $\bar{\gamma}_{inl}$ as shown in \eqref{eq:de_sinr} involves significant complexity since it requires solving a branch of fixed-point equations \eqref{eq:e} - \eqref{eq:T_il'}. Recall that the RZF precoding exploited in this work is essentially an approximation of the prevailing ZF precoding scheme with very small regularization factor $\delta$. Based on this fact, one may guess that the interference term of $\bar{\gamma}_{inl}$, i.e. the first term for the denominator of the right-hand side in \eqref{eq:de_sinr} converges to zero as $\delta$ goes to zero. To prove this, we first need the following additional assumptions.

\begin{assumption}
For arbitrary large $N_t$, the minimum eigenvalue of $\frac{1}{N_t} \sum_{i,l} \tilde{\mathbf{h}}_{inl} \tilde{\mathbf{h}}_{inl}^H$ is always bounded away from zero. \label{assumption:eign}
\end{assumption}

\begin{assumption}
For arbitrary large $N_t$, the limit $\bar{e}_{inl} = \lim_{\delta \to 0} \delta e_{inl} (\delta)$ always exists and $\bar{e}_{inl} > \epsilon$ for any $i$, $n$ and $l$ given some $\epsilon > 0$. \label{assumption:limit}
\end{assumption}

Under these two assumptions, the original deterministic equivalent term \eqref{eq:de_sinr} can be further approximated according to the following theorem.

\begin{theorem}
Let assumptions \ref{assumption:eign} and \ref{assumption:limit} hold true. As $\delta \to 0$, the deterministic equivalent expression \eqref{eq:de_sinr} converges to 
\begin{equation}
    \lim_{\delta \to 0} \bar{\gamma}_{inl}=\frac{N_{t}\alpha_{inl}p_{inl}\bar{e}_{inl}}{\sigma_{in}^{2}},
\end{equation}
where
\begin{equation}
    \bar{e}_{inl}  = \frac{1}{N_t} \mathrm{Tr}\left({\mathbf{R}}_{il}\left( \sum_{j,q}\frac{\alpha_{jnq}{\mathbf{R}}_{jq}}{N_t\bar{e}_{jnq}}+ \mathbf{I}  \right)^{-1} \right), \forall i,n,l. \label{eq:e_hat}
\end{equation} \label{theorem:de_sinr_zf}
\end{theorem}
\begin{proof}
The proof is given in appendix \ref{app:theorem_de_sinr_zf}.
\end{proof}

After deriving the deterministic equivalent of the SINR expression using Theorems \ref{theorem:DE_SINR} and \ref{theorem:de_sinr_zf}, the deterministic equivalent of $R_i$, denoted by $\bar{R}_i$, can be obtained by exploiting the so-called continuous mapping theorem \cite{ash2000probability} such that
\begin{equation}
    \bar{R}_{i} = \sum_{n,l} \log \left( 1 + \frac{N_{t}\alpha_{inl}p_{inl}\bar{e}_{inl}}{\sigma_{in}^{2}} \right). \label{eq:R}
\end{equation}

Once the deterministic equivalent term $\bar{R}_i$ is calculated, we substitute it along with the equality constraints \eqref{eq:e_hat} into problem \eqref{prblm:soft_appro} to replace the original rate expression $R_i$. The converted problem can thus be formulated as

\begin{subequations}
\begin{align}
   \mathop{\text{max}}_{\{\alpha_{inl}, p_{inl}, \bar{e}_{inl}\}} \ & \sum_{i=1}^I  \mu_i \bar{R}_i+w\sum_{i=1}^I \frac{1}{1+e^{-\theta(\bar{R}_i -R_i^{\min})}} \label{obj:app} 
  \\
 \text{s. t.} \quad \quad \ & \eqref{constr:str} - \eqref{constr:powge0}, \ \eqref{eq:e_hat}.
 \end{align}  \label{prblm:app}
\end{subequations}

\section{Algorithm Design}
In this section we propose efficient algorithm designs to solve the approximated optimization problem \eqref{prblm:app}. This problem is highly non-convex mainly due to the integer constraints \eqref{constr:int} and constraints \eqref{eq:e_hat} which are implicit fixed-point equation constraints w.r.t. variable $\bar{e}_{inl}$. In the following contexts, we first perform linear relaxation techniques with tricks \cite{zhou2021efficient} on constraints \eqref{constr:int} to make the problem become a continuous optimization problem whose solutions are naturally close to integers. Next we apply the so-called implicit optimization techniques \cite{travacca2020implicit} to deal with the fixed-point equation constraints, such that the gradient of the objective function can be explicitly expressed in closed form. With the derived gradient, we further develop an efficient algorithm design based on the prevailing accelerated Frank-Wolfe method \cite{doi:10.1057/palgrave.jors.2600425} to solve problem \eqref{prblm:app}.

\subsection{Linear Relaxation of the Optimization Problem}
To solve the mixed-integer non-convex problem (MINP) \eqref{prblm:app}, an efficient approach is to perform the linear relaxation to the integer constraints so that the optimization problem becomes continuous. After finishing solving the continuous problem, the ultimate solution for the original MINP can be derived by rounding the optimization results back into 0-1 integers. However, if we simply replace the integer constraints \eqref{constr:int} with their corresponding linear relaxations $0 \leq \alpha_{inl} \leq 1, \forall i,n,l$, it
may lead to an undesirable numerical solution (i.e. roughly equal to 0.5) for which naive rounding scheme would inevitably incur huge performance loss. To address this issue, we introduce an additional exponent $r$ to $\alpha_{inl}$ and rewrite the rate expression \eqref{eq:R} as
\begin{equation}
    \hat{R}_{i} = \sum_{n,l} \log \left( 1 + \frac{N_{t}\alpha_{inl}^{r}p_{inl}\bar{e}_{inl}}{\sigma_{in}^{2}} \right). 
\end{equation}
which is exactly identical to $R_i$ in \eqref{eq:R} since $\alpha_{inl}$ is Boolean. On the other hand, once constraints \eqref{constr:int} being relaxed, the existence of $r$ will avoid the case where the optimized $\alpha_{inl}$ can hardly be rounded, given that $r$ is sufficiently large. Specifically, if $\alpha_{inl}$ takes any value between 0 and 1, its $r$-th power will be nearly 0 given that $r$ is large enough. Note that the effectiveness of such trick for linear relaxation has been demonstrated by the authors in previous works for dealing with MINP \cite{zhou2021efficient,xue2018time}. After performing linear relaxation, the original problem becomes

\begin{subequations}
\begin{align}
   \mathop{\text{max}}_{\{\alpha_{inl}, p_{inl}, \bar{e}_{inl}\}} \ & \sum_{i=1}^I  \mu_i \hat{R}_i+w\sum_{i=1}^I \frac{1}{1+e^{-\theta(\hat{R}_i -R_i^{\min})}}  \label{obj:final}
  \\
 \text{s. t.}  \quad \quad & 0 \leq \alpha_{inl} \leq 1, \quad \forall i, n, l, \\
 \quad \quad \ & \eqref{constr:str}, \eqref{constr:pow}, \eqref{constr:powge0}, \eqref{eq:e_hat}.
 \end{align} \label{prblm:final}
\end{subequations}

\subsection{Implicit Gradient of the Optimization Problem}

Though being continuous, problem \eqref{prblm:final} is still non-tractable due to the implicit fixed-point constraints \eqref{eq:e_hat}, which prohibits the exploitation of conventional gradient-based algorithms. To overcome this issue,  we first need the following highly useful theorem, such that the gradient for the objective function \eqref{obj:final} can be computed explicitly. For notation simplicity, we denote $\boldsymbol{\alpha} = [\alpha_{111}, \cdots, \alpha_{inl}, \cdots, \alpha_{INL}]^T$ and $\bar{\boldsymbol{e}} = [\bar{e}_{111}, \cdots, \bar{e}_{inl}, \cdots, \bar{e}_{INL}]^T$ in the subsequent context.

\begin{theoremNoParens}[{\cite[Th. II.1.]{travacca2020implicit}}]
    Consider a continuously differentiable implicit map $\mathbf{x} = \phi (\mathbf{x}, \mathbf{u}) $ where $\mathbf{x} \in \mathbb{R}^{n}$, $\mathbf{u} \in \mathbb{R}^{m}$ and $\phi: \mathbb{R}^{n} \times \mathbb{R}^{m} \rightarrow \mathbb{R}^{n}$. Suppose that for any given  $\mathbf{u}$, $\mathbf{x}$ is uniquely defined by the implicit equation such that $\mathbf{x} = \mathbf{x} (\mathbf{u})$. In other words, $\mathbf{x}$ can be computed using the following fixed-point iteration:
    \begin{equation}
        \mathbf{x}^{k+1} (\mathbf{u}) = \phi \left( \mathbf{x}^k(\mathbf{u}), \mathbf{u} \right)
    \end{equation}
    starting with some $\mathbf{x}^0(\mathbf{u})$. Then the gradient of $\mathbf{x}(\mathbf{u})$ w.r.t. $\mathbf{u}$ can be written as
    \begin{equation}
        \nabla_{\mathbf{u}} \mathbf{x}(\mathbf{u}) = \nabla_{\mathbf{u}} \phi(\mathbf{x}, \mathbf{u}) \left( \mathbf{I} - \nabla_{\mathbf{x}} \phi(\mathbf{x}, \mathbf{u}) \right)^{-1} .
    \end{equation} \label{theorem:io}
\end{theoremNoParens}

According to Theorem \ref{theorem:de}, recall that  $\bar{e}_{inl}$ can be uniquely obtained through the fixed-point iteration based on \eqref{eq:e_hat} with initial value of $1$  for any given $\alpha_{inl}$. In other words,  we can represent $\bar{e}_{inl}$ as $\bar{e}_{inl} = \bar{e}_{inl} (\boldsymbol{\alpha})$.  Denote the objective function \eqref{obj:final} by $ F = F_1  + F_2 $, where
\begin{align}
    & F_{1}  \triangleq \sum_{i=1}^I \mu_i f_{i1} \left(\bar{e}_{inl} \left( \boldsymbol{\alpha}  \right),\alpha_{inl}\right), \\
    & F_{2} \triangleq w \sum_{i=1}^I f_{i2}\left(\bar{e}_{inl}\left( \boldsymbol{\alpha}  \right),\alpha_{inl}\right), \\
    & f_{i1} \left(\bar{e}_{inl} \left( \boldsymbol{\alpha}  \right),\alpha_{inl}\right) \triangleq \hat{R}_i, \\
    & f_{i2}\left(\bar{e}_{inl}\left( \boldsymbol{\alpha}  \right),\alpha_{inl}\right) \triangleq \frac{1}{1+e^{-\theta \left(\hat{R}_i -R_i^{\min} \right)}}. \\
\end{align}

 According to the chain rule and Theorem \ref{theorem:io}, the partial derivative of $f_{i1}$ with respect to ${\alpha}_{inl}$ can be written as

\begin{equation}
    \frac{\partial F_{1}}{ \partial {\alpha}_{inl} } = \mu_i \frac{\partial f_{i1} \left(\bar{e}_{inl},\alpha_{inl} \right) }{ \partial {\alpha}_{inl} }  +
    \sum_{j,q}\mu_j g_{jnq} \cdot
    \frac{ \partial \bar{e}_{jnq}(\boldsymbol{\alpha}) } { \partial \alpha_{inl} }  \label{eq:F1}
\end{equation}
where
\begin{align}
    & g_{jnq} = \frac{ \partial f_{j1} \left(\bar{e}_{jnq},\alpha_{jnq} \right)} { \partial \bar{e}_{jnq} } \left|_{\bar{e}_{jnq} =\bar{e}_{jnq}(\boldsymbol{\alpha})} \right. \\
    & \frac{\partial f_{i1} \left(\bar{e}_{inl},\alpha_{inl} \right) }{ \partial {\alpha}_{inl} }
    = \frac{r N_{t}\alpha^{(r-1)}_{inl} p_{inl}\bar{e}_{inl}}{\left(1+\hat{\gamma}_{inl} \right)\sigma_{i n}^{2}}, \\
    &  \frac{ \partial f_{i1} \left(\bar{e}_{inl},\alpha_{inl} \right)} { \partial \bar{e}_{inl} }
    =  \frac{N_{t}\alpha^r_{inl} p_{inl}}{\left( 1+\hat{\gamma}_{inl} \right) \sigma_{i n}^{2}}
\end{align}
with $\hat{\gamma}_{inl} = (N_{t}\alpha_{inl}^{r}p_{inl}\bar{e}_{inl}) / \sigma_{in}^{2}$. Therefore, it remains to derive the partial derivative $\partial \bar{e}_{jnq}(\boldsymbol{\alpha}) / \partial \alpha_{inl} $ such that \eqref{eq:F1} can be fully expressed. To address this issue, we first denote each constraint in \eqref{eq:e_hat} by $\bar{{e}}_{inl} = \Phi_{inl} (\bar{\mathbf{e}}, \boldsymbol{\alpha}), \ \forall i, n, l$, after which we rewrite all constraints in \eqref{eq:e_hat} into a more compact form as $\bar{\mathbf{e}} = \Phi (\bar{\mathbf{e}}, \boldsymbol{\alpha})$.
By applying Theorem \ref{theorem:io}, we have 
\begin{equation}
    \nabla_{\boldsymbol{\alpha}} \bar{\boldsymbol{e}}(\boldsymbol{\alpha}) = \nabla_{\boldsymbol{\alpha}} \Phi(\bar{\mathbf{e}},\boldsymbol{\alpha}) \left( \mathbf{I}-\nabla_{\bar{\mathbf{e}}}\Phi(\bar{\mathbf{e}},\boldsymbol{\alpha})\right) ^{-1} \label{eq:partial_deri_e}
\end{equation}
where $\nabla_{\boldsymbol{\alpha}} \bar{\boldsymbol{e}}(\boldsymbol{\alpha})$, $\nabla_{\boldsymbol{\alpha}} {\boldsymbol{\Phi}}(\bar{\mathbf{e}}, \boldsymbol{\alpha})$ and $\nabla_{\bar{\mathbf{e}}} {\boldsymbol{\Phi}}(\bar{\mathbf{e}},\boldsymbol{\alpha})$ are $(I\cdot N \cdot L)\times(I\cdot N \cdot L)$ gradient matrices whose $(i \cdot n \cdot l)-(j \cdot n \cdot q)$ elements are $\partial \bar{e}_{inl}(\boldsymbol{\alpha}) / \partial \alpha_{jnq} $, $\partial {\Phi}_{inl}(\bar{\mathbf{e}},\boldsymbol{\alpha}) / \partial \alpha_{jnq} $  and $\partial {\Phi}_{inl}(\bar{\mathbf{e}},\boldsymbol{\alpha}) / \partial \bar{e}_{jnq} $, respectively. By further exploiting the chain rule, we have
\begin{align}
\frac{\partial {\Phi}_{inl}(\bar{\mathbf{e}},\boldsymbol{\alpha})}{ \partial \alpha_{jnq} }
& = -\frac{\text{Tr}\left({\mathbf{R}}_{il}\mathbf{T}_{n}^{-1}{\mathbf{R}}_{jq}\mathbf{T}_{n}^{-1} \right)}{N_{t}^2 \bar{e}_{jnq}}
\\
     \frac{\partial {\Phi}_{inl}(\bar{\mathbf{e}},\boldsymbol{\alpha})}{ \partial \bar{e}_{jnq} } & =  \frac{\alpha_{jnq}\text{Tr}\left({\mathbf{R}}_{il}\mathbf{T}_{n}^{-1}{\mathbf{R}}_{jq}\mathbf{T}_{n}^{-1} \right)}{N_{t}^2 \bar{e}_{jnq}^2}. \label{eq:partial_deri_phi_e}
\end{align} 
Combining \eqref{eq:partial_deri_e} - \eqref{eq:partial_deri_phi_e}, $\nabla_{\boldsymbol{\alpha}} \bar{\boldsymbol{e}}(\boldsymbol{\alpha})$ can be explicitly computed, so as each of its elements $\partial \bar{e}_{inl}(\boldsymbol{\alpha}) / \partial \alpha_{jnq} $. This finishes the derivation of ${\partial F_{1}}/{ \partial {\alpha}_{inl} }$. 

As for the derivative w.r.t $f_{i2}$, note that $f_{i2}$ is a function w.r.t. $f_{i1}$, or in other words, $\hat{R}_{i}$, its partial derivative w.r.t. $\alpha_{inl}$ can be easily obtained according to the chain rule: 

\begin{equation}
    \frac{\partial F_{2}}{ \partial {\alpha}_{inl} } = w \left( m_{i} \frac{\partial f_{i1} }{ \partial {\alpha}_{inl} }  +
    \sum_{j,q} m_j g_{jnq} \cdot
    \frac{ \partial \bar{e}_{jnq}(\boldsymbol{\alpha}) } { \partial \alpha_{inl} } \right), \label{eq:F2}
\end{equation}
where
\begin{equation}
     m_i = \frac{\theta e^{-\theta\left(\hat{R}_i -R_i^{min}\right)}}{\left(1+e^{-\theta(\hat{R}_i -R_i^{min})}\right)^2} \label{eq:m}
\end{equation}
is the derivative of the outer sigmoid function. Combining \eqref{eq:F1} and \eqref{eq:F2} together we derive the partial derivative of \eqref{obj:final} w.r.t. $\alpha_{inl}$:
\begin{equation}
     \frac{\partial F}{ \partial {\alpha}_{inl} } = \frac{\partial F_{1}}{ \partial {\alpha}_{inl} } + \frac{\partial F_{2}}{ \partial {\alpha}_{inl} }. \label{eq:partial_deri_F_alpha}
\end{equation}

On the other hand, the derivation for the gradient w.r.t. $p_{inl}$ is rather straightforward, since the implicit constraints \eqref{eq:e_hat} does not involve $p_{inl}$. Specifically, the partial derivative of $F$ w.r.t. $p_{inl}$ can be computed by 
\begin{equation}
\frac{\partial F}{ \partial {p}_{inl} } = \frac{N_{t} \alpha^r_{inl}\bar{e}_{inl}(1+wm_i)}{\left( 1+\hat{\gamma}_{inl} \right) \sigma_{i n  }^{2}}. \label{eq:partial_deri_F_p}
\end{equation}

Let $\mathbf{p} = [p_{111}, \cdots, p_{inl}, \cdots, p_{INL}]^T$. 
With the help of \eqref{eq:partial_deri_F_alpha} and \eqref{eq:partial_deri_F_p}, the gradient of F w.r.t. $\boldsymbol{\alpha}$ and $\mathbf{p}$ can be respectively expressed as
\begin{align}
    \nabla_{\boldsymbol{\alpha}} F({\boldsymbol{\alpha}}) & = \left[ \frac{\partial F}{ \partial {\alpha}_{111} }, \cdots, \frac{\partial F}{ \partial {\alpha}_{inl} },\cdots, \frac{\partial F}{ \partial {\alpha}_{INL} } \right], \\
    \nabla_{\mathbf{p}} F({\mathbf{p}}) & = \left[ \frac{\partial F}{ \partial {p}_{111} }, \cdots, \frac{\partial F}{ \partial {p}_{inl} },\cdots, \frac{\partial F}{ \partial {p}_{INL} } \right].
\end{align}

\subsection{Accelerated Frank-Wolfe Method}
To tackle the coupling of the optimization variables $\alpha_{inl}$ and $p_{inl}$, a natural solution \eqref{prblm:final} is to employ alternating optimization techniques, which decomposes the original problem into the following two subproblems:
\begin{subequations}
\begin{align}
   \mathop{\text{max}}_{\{\boldsymbol{\alpha}, \bar{\mathbf{e}}\}} \ & \sum_{i=1}^I  \mu_i \hat{R}_i+w\sum_{i=1}^I \frac{1}{1+e^{-\theta(\hat{R}_i -R_i^{\min})}} 
  \\
 \text{s. t.}  \ \  & 0 \leq \alpha_{inl} \leq 1, \quad \forall i, n, l, \label{constr:alpha_1} \\
    & \sum_{i, l} \alpha_{inl} \leq L_{\text{max}}, \quad \forall \ n, \label{constr:alpha_2} \\
 & \bar{e}_{inl}  = \frac{1}{N_t} \mathrm{Tr}\left({\mathbf{R}}_{il}\left( \sum_{j,q}\frac{\alpha_{jnq}{\mathbf{R}}_{jq}}{N_t\bar{e}_{jnq}}+ \mathbf{I}  \right)^{-1} \right), \forall i,n,l. 
 \end{align} \label{subprblm:alpha}
\end{subequations}
and
\begin{subequations}
\begin{align}
   \mathop{\text{max}}_{\{ {p}_{inl}\}} \ & \sum_{i=1}^I  \mu_i \hat{R}_i+w\sum_{i=1}^I \frac{1}{1+e^{-\theta(\hat{R}_i -R_i^{\min})}} 
  \\
 \text{s. t.}  \ \  
    & \sum_{i,n,l} p_{inl} \leq P_{\text{max}}, \label{constr:p_1} \\
    &  p_{inl} \geq 0, \quad \forall i, n, l, \label{constr:p_2}
 \end{align} \label{subprblm:p}
\end{subequations}
respectively. Problem \eqref{subprblm:alpha} and \eqref{subprblm:p} are alternately optimized while holding the other block of variables fixed. Since the gradient w.r.t.  $\alpha_{inl}$ and $p_{inl}$ have already derived explicitly in the previous subsection, it is reminiscent to exploit gradient-based algorithms to solve the aforementioned two subproblems, such as projected gradient descent \cite{nesterov2003introductory}, projected accelerated gradient method \cite{allen2014linear} or Frank-Wolfe method \cite{frank1956algorithm}. Among these kinds of algorithms, FW-based algorithms are particularly suitable for our considered problems \eqref{subprblm:alpha} and \eqref{subprblm:p}. The reason lies that these two subproblems only involves simple linear constraints, or more specifically, (capped) simplex constraints, such as constraints \eqref{constr:alpha_1} and \eqref{constr:alpha_2}, or constratints \eqref{constr:alpha_1}
respectively. Problem \eqref{subprblm:alpha} and \eqref{subprblm:p} are alternately optimized while holding the other block of variables fixed. Since the gradient w.r.t.  $\alpha_{inl}$ and $p_{inl}$ have already derived explicitly in the previous subsection, it is reminiscent to exploit gradient-based algorithms to solve the aforementioned two subproblems, such as projected gradient descent \cite{nesterov2003introductory}, projected accelerated gradient method \cite{allen2014linear} or Frank-Wolfe method \cite{frank1956algorithm}. Among these kinds of algorithms, FW-based algorithms are particularly suitable for our considered problems \eqref{subprblm:alpha} and \eqref{subprblm:p}. The reason lies that these two subproblems only involves simple linear constraints after calculating (implicit) gradients, or more specifically, (capped) simplex constraints, such as constraints \eqref{constr:alpha_1} and \eqref{constr:alpha_2}, or constraints \eqref{constr:p_1} and \eqref{constr:p_2}. This makes the optimization problem at each iteration of FW-based algorithms rather easy to solve, or even optmized in closed form.

However,  conventional Frank-Wolfe method comes with reduced complexity at the sacrifice of slow convergence. To overcome this drawback, the authors in \cite{li2020does} proposed an accelerated Frank-Wolfe (AFW) method. The key idea of AFW is to combine Frank-Wolfe method with the momentum, so that it takes the advantages of both the easier implementation of FW algorithms as well as the fast convergence of accelerated gradient descent methods. Therefore, we finally adopt this AFW method in our work to solve problem \eqref{prblm:final}, the algorithm of which can be summarized as follows, where $\mathbf{y}$, $\boldsymbol{\lambda}$, $\mathbf{v}$ and $\mathbf{z}$, $\boldsymbol{\nu}$, $\mathbf{w}$ are auxiliary variables w.r.t. $\boldsymbol{\alpha}$ and $\mathbf{p}$ respectively, $\eta^{(k)}$ is the ascent step size at the $k$-th iteration.

\begin{figure}[htb]
  \centering
\begin{minipage}{.95\linewidth}
\begin{algorithm}[H]
	\caption{Once-and-for-all multi-user scheduling algorithm with statistical CSI based on AFW \vspace{0.1cm} }
	\begin{algorithmic}[1]
	 \STATE Initialize $\boldsymbol{\alpha}$, $\mathbf{v}$, $\mathbf{p}$ and $\mathbf{w}$, $\boldsymbol{\lambda}^{(0)}=0$, $\boldsymbol{\nu}^{(0)}=0$
		\FOR {$ k =0,1,\cdots,K$} 
		\STATE  $\mathbf{y}^{(k)}=\boldsymbol{\alpha}^{(k)}+\eta^{(k)}(\mathbf{v}^{(k)}-\boldsymbol{\alpha}^{(k)})$;
		\STATE 
		Compute the gradient of $\mathbf{y}^{(k)}$: $\nabla_{\boldsymbol{\alpha}} F(\mathbf{y}^{(k)})$;
		\STATE 
		$ \boldsymbol{\lambda}^{(k+1)}=\boldsymbol{\lambda}^{(k)}+\eta^{(k)}(\nabla F(\mathbf{y}^{(k)})-\boldsymbol{\lambda}^{(k)})$;
		\STATE
		$\mathbf{v}^{(k+1)}=\text{arg max}_{\boldsymbol{\alpha}}\left<\boldsymbol{\lambda}^{(k+1)},\boldsymbol{\alpha} \right>$;
		\STATE
		$\boldsymbol{\alpha}^{(k+1)}=\boldsymbol{\alpha}^{(k)}+\eta^{(k)}(\mathbf{v}^{(k+1)}-\boldsymbol{\alpha}^{(k)})$;
		\STATE  $\mathbf{z}^{(k)}=\mathbf{p}^{(k)}+\eta^{(k)}(\mathbf{w}^{(k)}-\mathbf{p}^{(k)})$;
		\STATE 
		Compute the gradient of $\mathbf{z}^{(k)}$: $\nabla_{\mathbf{p}} F(\mathbf{z}^{(k)})$;
		\STATE 
		$ \boldsymbol{\nu}^{(k+1)}=\boldsymbol{\nu}^{(k)}+\eta^{(k)}(\nabla F(\mathbf{z}^{(k)})-\boldsymbol{\nu}^{(k)})$;
		\STATE
		$\mathbf{w}^{(k+1)}=\text{arg max}_{p}\left<\boldsymbol{\nu}^{(k+1)},\mathbf{p} \right>$;
		\STATE
		$\mathbf{p}^{(k+1)}=\mathbf{p}^{(k)}+\eta^{(k)}(\mathbf{w}^{(k+1)}-\mathbf{p}^{(k)})$;
		\ENDFOR
    \STATE
    \RETURN  $\boldsymbol{\alpha}^{(K)}$ and $\mathbf{p}^{(K)}$  
	\end{algorithmic} 
	\label{alg:AFW}
\end{algorithm}
\end{minipage}
\end{figure}

\section{Numerical Results}
In this section, we present various numerical results to
evaluate the performance of our proposed algorithm. Throughout the simulations, the number of receive antennas $Nr$ and the maximum number of data streams for each user $L$ are set to be 4 and 2, respectively, unless otherwise stated. Each user’s weight $\mu_i$ is equal to 1. The maximum transmit power at the BS is $P_{\max} = 10$ Watt. The regularization coefficient of RZF precoding is $10^{-10}$. Two kinds of channels are considered throughout the simulations: the first one is Rayleigh channel fading with identity channel covariance matrices, while the second is computed with the help of the \emph{Quadriga} wireless channel simulator.

\subsection{Effectiveness of Deterministic Equivalent}
In this subsection we focus on verifying the effectiveness of the proposed approximated rate expression \eqref{eq:R} based on deterministic equivalent. Without loss of generality, the number of schedulable RBG is set to one throughout this subsection. 

 We first investigate how accurate the proposed sum rate approximation based on deterministic equivalent with varying BS antennas, while holding the transmit SNR fixed as 10 dB. Define the relative error between real sum rate and sum rate with deterministic equivalent as
 \begin{equation}
     \epsilon = \frac{\sum_{i}\mu_iR_i - \sum_{i} \mu_i \bar{R}_i}{\sum_{i} \mu_i \bar{R}_i}.
 \end{equation}
 The simulation results presented in Fig.  \ref{fig:verifi_of_de_vs_BS_antenna} depict the average relative error of the sum rate approximation compared to the real sum rate with varying number of BS antennas, averaged over 1000 independent channel realizations, when the number of users are set to 4 and 8, respectively. From Fig.  \ref{fig:verifi_of_de_vs_BS_antenna}, we observe that the approximated sum rate becomes more accurate with increasing number of antennas under both Rayleigh channel and Quadriga-generated channel. Not surprisingly, the average relative error of the case with Rayleigh channel is smaller than than with Quadriga-generated channel, thanks to the elegant formulation of Rayleigh channels.

\begin{figure}[htp]
    \centering
    \includegraphics[width=1.0\linewidth]{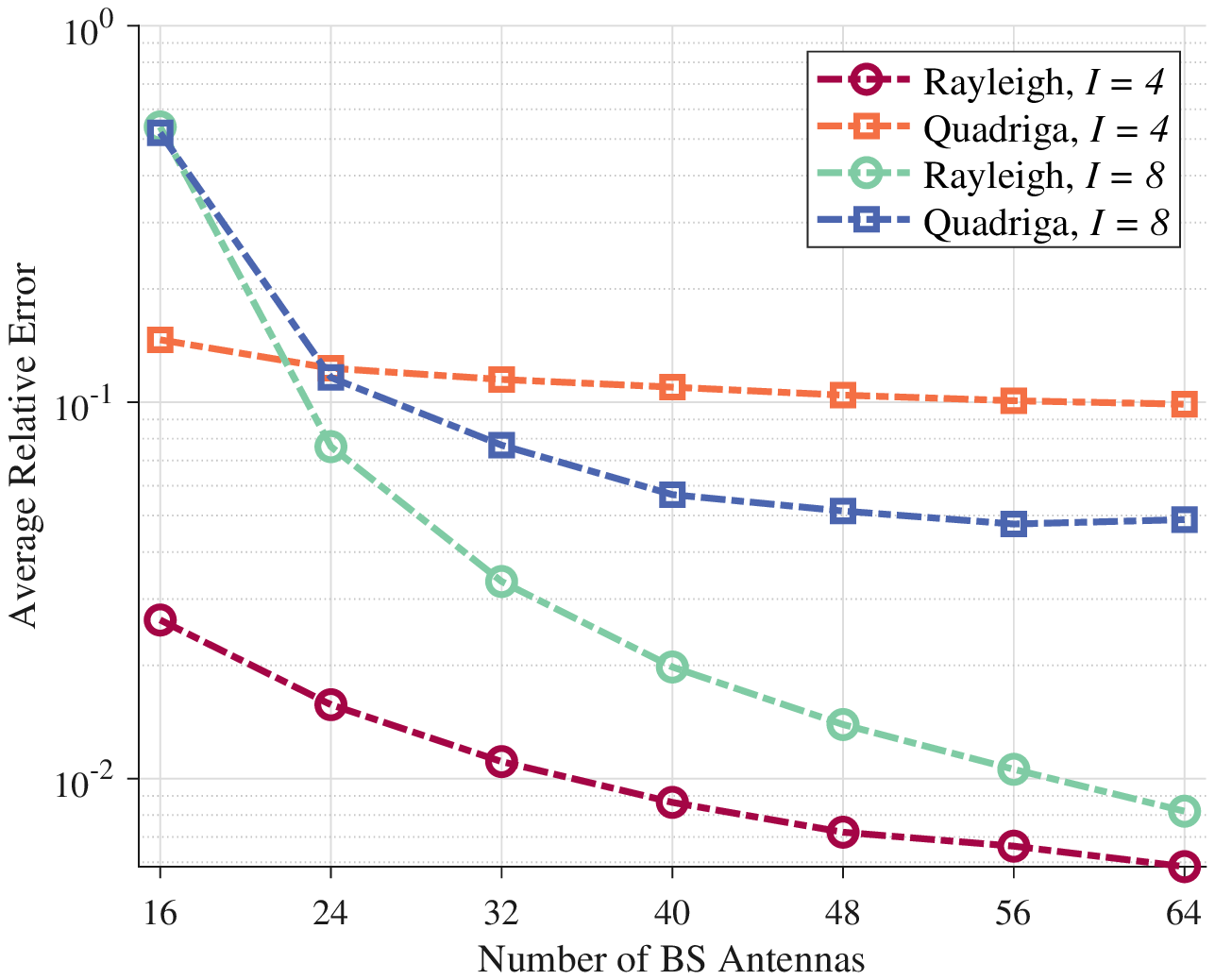}
    \caption{Relative error between real sum rate and sum rate with D.E. versus varying number of BS antennas. }
    \label{fig:verifi_of_de_vs_BS_antenna}
\end{figure}

Fig. \ref{fig:verifi_of_de_vs_snr} compares the sum rate computed by various schemes with varying transmit SNR under the case of 64 BS antennas, when the number of users are 4 and 8 respectively. It is revealed that the deterministic equivalent expression is an accurate approximation to the real rate expression for varying SNR. We conclude that the approximations are accurate even for small dimensions and can be applied to various optimization problems discussed in the sequel.

\begin{figure}[htp]
    \centering
    \includegraphics[width=1.0\linewidth]{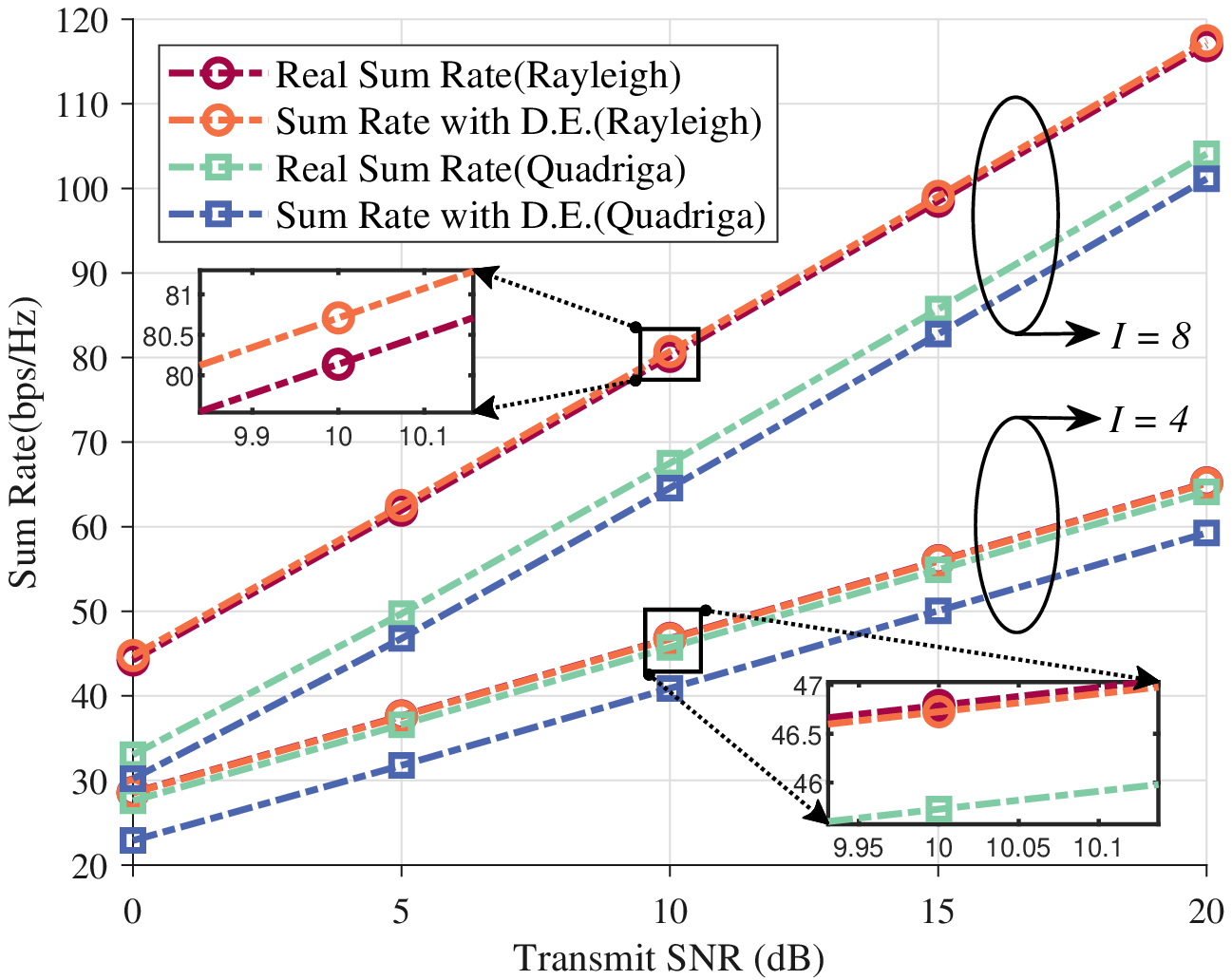}
    \caption{Sum rate versus varying transmit SNR.  }
    \label{fig:verifi_of_de_vs_snr}
\end{figure}

\subsection{Convergence of AFW}

We then examine the convergence of the proposed AFW-based algorithm.  Without loss of generality, only Rayleigh fading channel is considered in this subsection. The number of schedulable RBG is still set to one. We consider the case with 4 users, while the transmit SNR is fixed to be 10 dB. The parameters $r$, $\theta$ and $R_i^{\text{min}}$ are set to be 2, 5, and 5, respectively.

Fig. \ref{fig:ObjVal_vs_Iter} shows the average objective function value of the proposed AFW-based algorithm with 100 channel realizations, versus the number of iterations with varying penalty coefficients $w$. It is revealed from Fig. \ref{fig:ObjVal_vs_Iter} that the proposed AFW-based algorithm can converge rapidly, or in other words, in a few dozens times, under all of these cases.

\begin{figure}[htp]
    \centering
    \includegraphics[width=1.0\linewidth]{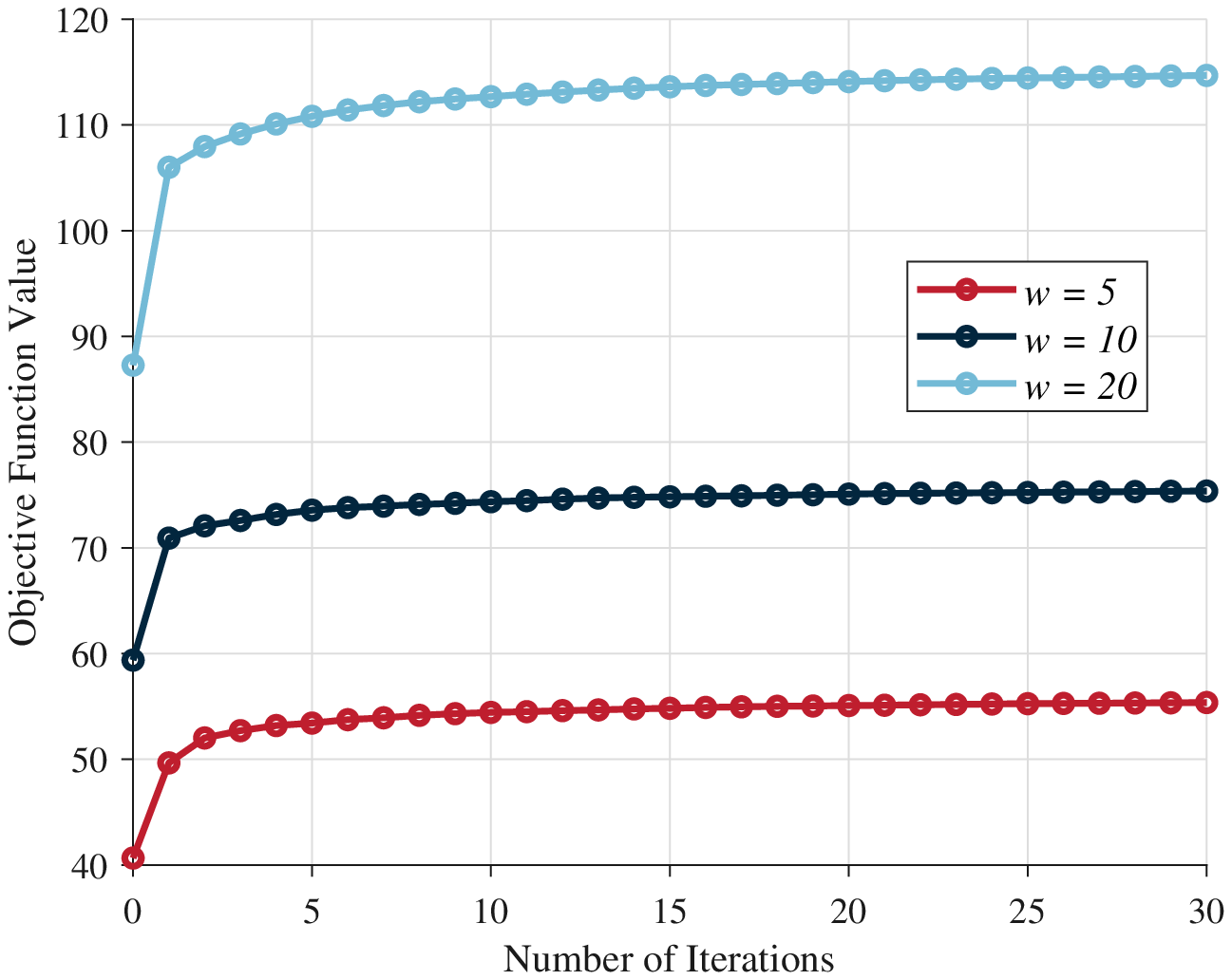}
    \caption{Average objective value versus the number of iterations.  }
    \label{fig:ObjVal_vs_Iter}
\end{figure}

\subsection{Performance Evaluation}

In this subsection, we evaluate the sum rate performance of the proposed algorithm versus varying transmit SNR. To verify the effectiveness of our proposed algorithm, we compare it with two benchmarks: the exhaustive search as well as the semi-orthogonal user selection (SUS) method. Since the exhaustive search faces with the \emph{curse of dimensional}, we consider a  simple case with 16 BS antennas and 4 users with Rayleigh fading, while each user only have one data stream. The parameters $r$, $\theta$, $w$ and $R_i^{\text{min}}$ are set to be 2, 5, 20 and 5, respectively. Fig. \ref{fig:ObjVal_vs_snr_4ue} shows the sum rate performance of various algorithms versus transmit SNR. On the one hand, it is revealed that our proposed algorithm achieve almost the same performance with that achieved by exhaustive search, which shows the effectiveness of our proposed algorithm. On the other hand, our algorithm significantly outperforms the SUS method, which is currently one of the most prevailing multi-user scheduling schemes exploited in mobile communication systems. 

\begin{figure}[htp]
    \centering
    \includegraphics[width=1.0\linewidth]{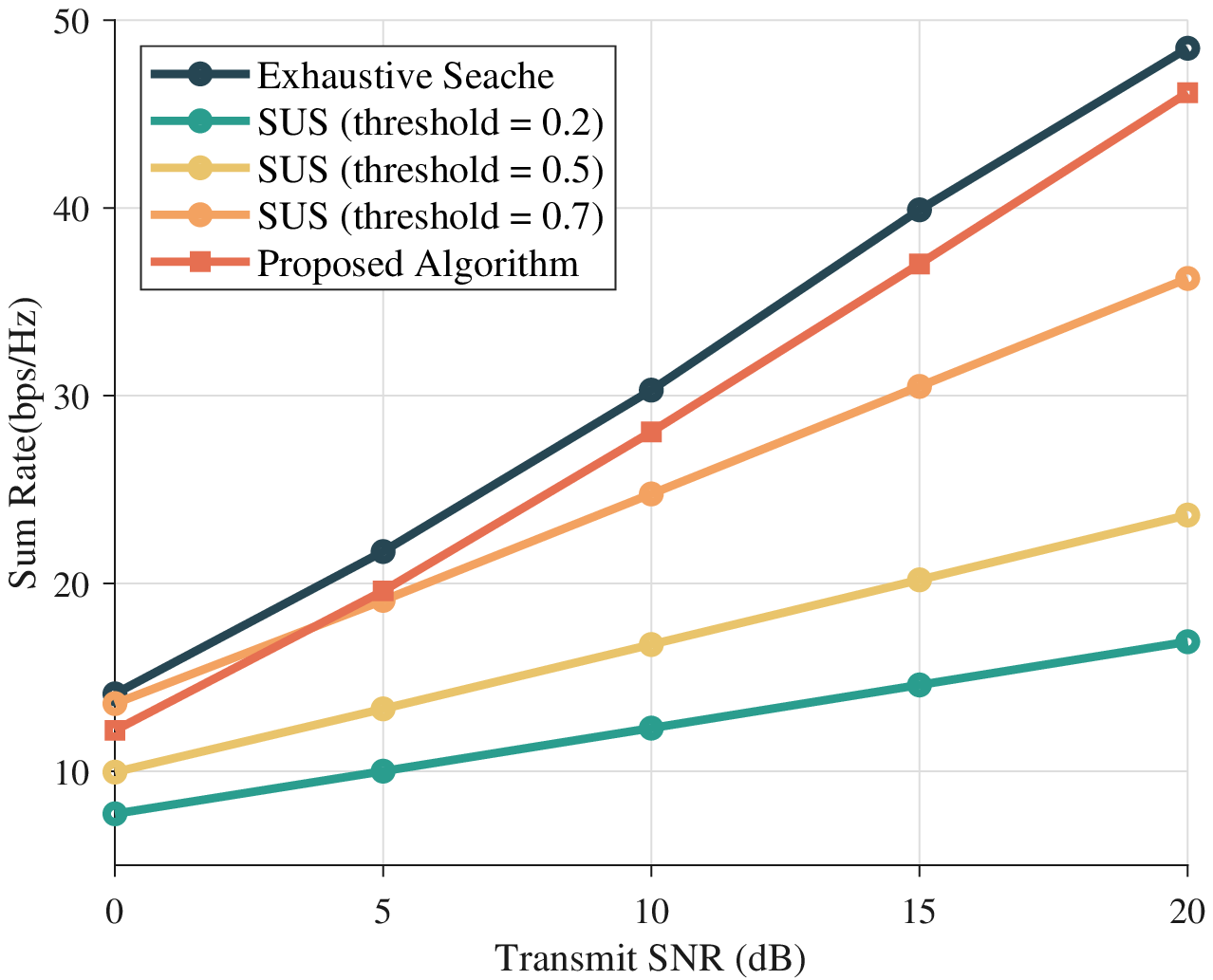}
    \caption{Sum rate versus transmit SNR.  }
    \label{fig:ObjVal_vs_snr_4ue}
\end{figure}


\section{Conclusion}
In this paper, we have investigated a long-term multi-user scheduling scheme in the MU-MIMO system based on statistical CSI. Rather than making unsatisfactory scheduling decisions at each TTI individually, the proposed scheme aims at giving a high-quality scheduling decision once and for all. With the help of deterministic equivalent, we formulate the ergodic sum rate optimization problem into a deterministic problem with fixed-point equality constraints. We the solve the problem via implicit gradient techniques and accelerated Frank-Wolfe method. The effectiveness of the proposed scheme have been verified through extensive simulations. The proposed scheduling scheme, along with their corresponding optimization algorithms, are probably better candidates than the conventional TTI-level counterpart, which may change the conventional mindset to deal with multi-user scheduling.


%

\appendices
\section{Useful Lemmas}

\begin{lemma}
    Resolvent Identity: Let $\mathbf{U}$ and $\mathbf{V}$ be two invertible complex matrices of size $N \times N$. Then
    \begin{equation}
        \mathbf{U}^{-1} - \mathbf{V}^{-1} = -\mathbf{U}^{-1}\left(\mathbf{U} - \mathbf{V} \right) \mathbf{V}^{-1}   
    \end{equation}
\end{lemma}

\begin{lemma}
 (Matrix Inversion Lemma): Let $\mathbf{U}$ be an $N \times N$ invertible matrix and $\mathbf{x} \in \mathbb{C}^N$, $c \in \mathbf{C}$ for which $\mathbf{U} + c \mathbf{x}\mathbf{x}^H$ is invertible. Then 
 \begin{equation}
    \mathbf{x}^{H}(\mathbf{U}+c\mathbf{x}\mathbf{x}^H)^{-1}=\frac{\mathbf{x}^H\mathbf{U}^{-1}}{1+c\mathbf{x}^H\mathbf{U}^{-1}\mathbf{x}} 
\end{equation}  \label{lemma:matrix_inversion}
\end{lemma}

\begin{lemma}
    Let $\mathbf{A}_1, \mathbf{A}_2,\cdots$ with $\mathbf{A}_N \in C^{N \times N}$, be a series of random matrices generated by the probability space $\left( \Omega, \mathcal{F}, P \right)$, such that for $\omega \in A \subset \Omega$, with $P(A) = 1$, $|| \mathbf{A}_N (\omega) || < K(\omega) < \infty$, uniformly on $N$. Let $\mathbf{x}_1, \mathbf{x}_2, \cdots,$ with $\mathbf{x}_N \in \mathbb{C}^{}$, be random vectors of i.i.d. entries with zero mean, variance $1/N$ and eighth order moment of order $O(1/N^4)$, independent of $\mathbf{A}_N$. Then \begin{equation}
    \mathbf{x}^H_N\mathbf{A}_N\mathbf{x}_N-\frac{1}{N}\text{Tr}\mathbf{A}_N\stackrel{N\to\infty}{\longrightarrow}0 
    \end{equation} \label{lemma:trace}
\end{lemma}

\begin{lemma}
    Let $\mathbf{A}_1, \mathbf{A}_2,\cdots$ with $\mathbf{A}_N \in C^{N \times N}$, be deterministic with uniformly bounded spectral norm and $\mathbf{B}_1, \mathbf{B}_2,\cdots,$ with $\mathbf{B}_N \in \mathbb{C}^{N \times N}$, be random Hermitian, with eigenvalues $\lambda_1^{\mathbf{B}_N} \leq \cdots \leq \lambda_N^{\mathbf{B}_N}$ such that, with probability one, there exists $\epsilon > 0$ for which $\lambda_1^{\mathbf{B}_N} > \epsilon$ for all large $N$. Then for $\mathbf{v} \in \mathbb{C}^{N}$
    \begin{equation}
        \frac{1}{N}\text{Tr}(\mathbf{A}_N\mathbf{B}^{-1}_N)-  \frac{1}{N}\text{Tr}[\mathbf{A}_N(\mathbf{B}_N+\mathbf{v}\mathbf{v}^H)^{-1}]\stackrel{N\to\infty}{\longrightarrow}0 
    \end{equation} \label{lemma:rand_one_perturbation}
    almost surely, where $\mathbf{B}_N^{-1}$ and $(\mathbf{B}_N + \mathbf{v}\mathbf{v}^H)^{-1}$ exist with probability one.
\end{lemma}

\section{Proof of Theorem \ref{theorem:DE_SINR}} \label{app:theorem_de_sinr}
Notice that $\gamma_{inl}$ is a function of the random channel vectors $\tilde{\mathbf{h}}_{inl}$. To be more specific, $\gamma_{inl}$ is composed of the three following terms w.r.t. $\tilde{\mathbf{h}}_{inl}$: 1) the power normalization coefficient $\xi_{inl}$; 2) the scaled signal power term $ \alpha_{inl} p_{inl} \xi_{inl}^2 |\tilde{\mathbf{h}}_{inl}^{H} {\mathbf{W}}_{nm} \tilde{\mathbf{h}}_{inl} |^{2}$; 3) the scaled interference power term $\sum_{(j, q) \neq(i, l)} \alpha_{jnq} p_{jnq} \xi_{jnq}^2 |\tilde{\mathbf{h}}_{inl}^{H} {\mathbf{W}}_{nm} \tilde{\mathbf{h}}_{jnq}|^{2}$. In the following context, we will first individually derive the deterministic equivalent (D.E.) expressions for each of the three aforesaid terms. These deterministic equivalent expressions are then further combined to constitute the final deterministic equivalent expression for $\gamma_{inl}$.

\subsection{D.E. of Normalization Coefficient} \label{subsec:norm_coff}
We first show the derivation for the deterministic equivalent of $\xi_{inl}$. Plugging \eqref{eq:chanl_mdl} into \eqref{eq:rzf_coff} and using lemma \ref{lemma:matrix_inversion}, we can rewrite $\xi_{inl}^2$ as
\begin{equation}
\begin{aligned}
    \xi_{inl}^2 & = \frac{N_t}{ {\mathbf{x}}_{inl}^H\tilde{\mathbf{R}}_{il}^{1/2} \mathbf{W}_{n}^{-2} \tilde{\mathbf{R}}_{il}^{1/2}{\mathbf{x}}_{inl}} \\
    & = \frac{N_t \left(1+\alpha_{inl}{\mathbf{x}}_{inl}^H\tilde{\mathbf{R}}_{il}^{1/2}\mathbf{W}^{-1}_{n[il]} \tilde{\mathbf{R}}_{il}^{1/2}{\mathbf{x}}_{inl}\right)^2}
    { {\mathbf{x}}_{inl}^H\tilde{\mathbf{R}}_{il}^{1/2} \mathbf{W}_{n[il]}^{-2} \tilde{\mathbf{R}}_{il}^{1/2}{\mathbf{x}}_{inl}}
\end{aligned} \label{eq:equiv}
\end{equation}
where
\begin{equation}
    \mathbf{W}_{n[il]} = \sum_{(j,q) \neq (i,l)} \alpha_{jnq} \tilde{\mathbf{h}}_{jnq} \tilde{\mathbf{h}}_{jnq}^{H}+\delta \mathbf{I} \text{.} 
\end{equation}
Then lemma \ref{lemma:trace} and lemma \ref{lemma:rand_one_perturbation} are subsequently employed to deal with both the numerator and the denominator for the right-hand side of the second equation in \eqref{eq:equiv} such that
\begin{align}
    & \xi_{inl}^2-\frac{\left(N_t+\alpha_{inl}\text{Tr}\left(\tilde{\mathbf{R}}_{il}\mathbf{W}_{n[il]}^{-1}\right)\right)^{2}}{\text{Tr}\left(\tilde{\mathbf{R}}_{il}\mathbf{W}_{n[il]}^{-2}\right)} \mathop{\longrightarrow}_{N_t\to\infty} 0 \\
    \Leftrightarrow \ 
    & \xi_{inl}^2-\frac{N_t\left(1+\frac{1}{N_t}\alpha_{inl}\text{Tr}\left(\tilde{\mathbf{R}}_{il}\mathbf{W}_{n}^{-1}\right)\right)^{2}}{\frac{1}{N_t}\text{Tr}\left(\tilde{\mathbf{R}}_{il}\mathbf{W}_{n}^{-2}\right)} \mathop{\longrightarrow}_{N_t\to\infty} 0 \text{.} \label{eq:after_rank_one_perturbation}
\end{align}
By exploiting theorem \ref{theorem:DE_SINR}, we have
\begin{equation}
    \frac{1}{N_t}\text{Tr}\left(\tilde{\mathbf{R}}_{il}\mathbf{W}_{n}^{-1}\right) - e_{inl} \mathop{\longrightarrow}_{N_t\to\infty} 0  \label{eq:xi_numerator}
\end{equation}
Taking the derivative of \eqref{eq:xi_numerator} w.r.t. $-\delta$, we also obtain
\begin{equation}
    \frac{1}{N_t}\text{Tr}\left(\tilde{\mathbf{R}}_{il}\mathbf{W}_{n}^{-2}\right) - e'_{inl} \mathop{\longrightarrow}_{N_t\to\infty} 0  \label{eq:xi_denominator}
\end{equation}
Combining \eqref{eq:after_rank_one_perturbation} - \eqref{eq:xi_denominator} together, we finally derive
\begin{equation}
\xi_{inl}^2-\frac{N_{t}(1+\alpha_{inl}e_{inl})^2}{e'_{inl}} \mathop{\longrightarrow}_{{N_t\to\infty}} 0 
\end{equation}
where $e_{inl}$ and $e'_{inl}$ can be calculated through equations \eqref{eq:e}, \eqref{eq:T} and \eqref{eq:e'}, \eqref{eq:T'}, respectively. This completes the proof for the deterministic equivalent of $\xi_{inl}^2$.

\subsection{D.E. of Signal Power Term}
In this subsection we deal with the signal power term $ \alpha_{inl} p_{inl} \xi_{inl}^2 |\tilde{\mathbf{h}}_{inl}^{H} {\mathbf{W}}_{n}^{-1} \tilde{\mathbf{h}}_{inl} |^{2}$. Since the deterministic equivalent of $\xi_{inl}^2$ has been derived in the previous subsection, we focus on the derivation for the  deterministic equivalent of the term $\tilde{\mathbf{h}}_{inl}^{H} {\mathbf{W}}_{n}^{-1} \tilde{\mathbf{h}}_{inl}$ in the following contexts. 

Through lemma \ref{lemma:matrix_inversion} we have
\begin{equation}
    \tilde{\mathbf{h}}_{inl}^{H} {\mathbf{W}}_{n}^{-1}\tilde{\mathbf{h}}_{inl}
    =
    \frac{{\mathbf{x}}_{inl}^H\tilde{\mathbf{R}}_{il}^{1/2} \mathbf{W}_{n[il]}^{-1} \tilde{\mathbf{R}}_{il}^{1/2}{\mathbf{x}}_{inl}}
    {1+\alpha_{inl}{\mathbf{x}}_{inl}^H\tilde{\mathbf{R}}_{il}^{1/2}\mathbf{W}^{-1}_{n[il]} \tilde{\mathbf{R}}_{il}^{1/2}{\mathbf{x}}_{inl}}.
    \label{eq:signal_power}
\end{equation}
Similar to the proofs in subsection \ref{subsec:norm_coff}, we adopt
lemma \ref{lemma:trace} and lemma \ref{lemma:rand_one_perturbation} such that the right-hand side of \eqref{eq:signal_power} can be further expressed as 
\begin{align}
    & \tilde{\mathbf{h}}_{inl}^{H} {\mathbf{W}}_{n}^{-1} \tilde{\mathbf{h}}_{inl} - \frac{\frac{1}{N_{t}}\text{Tr}\left(\tilde{\mathbf{R}}_{il}\mathbf{W}_{n[il]}^{-1}\right)}{1+\frac{1}{N_{t}}\alpha_{inl}\text{Tr}\left(\tilde{\mathbf{R}}_{il} \mathbf{W}_{n[il]}^{-1}\right)} \mathop{\longrightarrow}_{N_t \rightarrow \infty} 0 \\
    \Leftrightarrow & \tilde{\mathbf{h}}_{inl}^{H} {\mathbf{W}}_{n}^{-1} \tilde{\mathbf{h}}_{inl} -  \frac{\frac{1}{N_{t}}\text{Tr}\left(\tilde{\mathbf{R}}_{il}
    \mathbf{W}^{-1}_{n}\right)}{1+\frac{1}{N_{t}}\alpha_{inl}\text{Tr}\left(\tilde{\mathbf{R}}_{il} \mathbf{W}^{-1}_{n}\right)} \mathop{\longrightarrow}_{N_t \rightarrow \infty} 0.
\end{align}
By further combining with theorem \ref{theorem:DE_SINR} we conclude
\begin{equation}
\tilde{\mathbf{h}}_{inl}^{H} {\mathbf{W}}_{n}^{-1} \tilde{\mathbf{h}}_{inl} -
\frac{e_{inl}}{1+\alpha_{inl}e_{inl}}\mathop{\longrightarrow}_{N_t \rightarrow \infty} 0 
\end{equation}
where $e_{inl}$ is exactly the same solution computed by \eqref{eq:e} and \eqref{eq:T} in the previous subsection. This completes the proof for the deterministic equivalent of $\tilde{\mathbf{h}}_{inl}^{H} \tilde{\mathbf{W}}_{n}^{-1} \tilde{\mathbf{h}}_{inl}$.

\subsection{D.E. of Interference Power Term}
We then turn to the deterministic equivalent of the interference term $\sum_{(j, q) \neq(i, l)} \alpha_{jnq} p_{jnq} \xi_{jnq}^2 |\tilde{\mathbf{h}}_{inl}^{H} {\mathbf{W}}_{n} \tilde{\mathbf{h}}_{jnq}|^{2}$. Note that the deterministic equivalent of $\xi_{jnq}^2$, denoted by $\bar{\xi}_{jnq}^2$, has already been obtained in \ref{subsec:norm_coff}. Therefore, in the following contexts we replace $\xi_{jnq}^2$ with $\bar{\xi}_{jnq}^2$ and treat it as constant.

We first define $\beta_{inl} \triangleq \alpha_{inl} p_{inl}\hat{\xi}_{inl}^2$ and rewrite the interference term as shown in \eqref{eq:interf_eq} - \eqref{eq:interf_coff} at the top of next page.

\newcounter{de_interf}
\begin{figure*}[!t]
\normalsize
\setcounter{de_interf}{\value{equation}}
\setcounter{equation}{50}

\begin{align}
& \sum_{(j, q) \neq(i, l)} \alpha_{jnq} p_{jnq} \bar{\xi}_{jnq}^2 \left|\tilde{\mathbf{h}}_{inl}^{H} {\mathbf{W}}_{n}^{-1} \tilde{\mathbf{h}}_{jnq} \right|^{2} 
= 
\tilde{\mathbf{h}}_{inl}^{H} {\mathbf{W}}_{n}^{-1}
\tilde{\mathbf{H}}_{n[il]}^H\mathbf{B}_{n[il]}\tilde{\mathbf{H}}_{[il]}{\mathbf{W}}_{n}^{-1}
\tilde{\mathbf{h}}_{inl} \label{eq:interf_eq} \\ 
& \tilde{\mathbf{H}}_{n[il]}\triangleq\left[\tilde{\mathbf{h}}
_{1n1},\cdots,\tilde{\mathbf{h}}_{(i-1)nl},\cdots, \tilde{\mathbf{h}}_{in(l-1)},\tilde{\mathbf{h}}_{in(l+1)},\cdots,\tilde{\mathbf{h}}_{(i+1)nl},\cdots, \tilde{\mathbf{h}}_{inl}\right]^H \label{eq:interf_H} \\
& \mathbf{B}_{n[il]}\triangleq \text{diag}\left(\beta_{1n1}, \cdots, \beta_{(i-1)nl}, \cdots, \beta_{in(l-1)}, \beta_{in(l+1)}, \cdots, \beta_{(i+1)nl}, \cdots, \beta_{inl} \right) \label{eq:interf_coff}
\end{align} 

\hrulefill

\vspace*{4pt}
\end{figure*}

By defining $\mathbf{C}_{n[il]} \triangleq \tilde{\mathbf{H}}_{n[il]}^H\mathbf{B}_{n[il]}\tilde{\mathbf{H}}_{n[il]}$, we have the following relationships:
\begin{equation}
\begin{aligned}
&\tilde{\mathbf{h}}_{inl}^{H} {\mathbf{W}}_{n}^{-1} \mathbf{C}_{n[il]}
{\mathbf{W}}_{n}^{-1}
\tilde{\mathbf{h}}_{inl}
\\=
&\frac{1}{N_{t}}{\mathbf{x}}_{inl}^H\tilde{\mathbf{R}}_{il}^{1/2}\mathbf{W}^{-1}_{n}\mathbf{C}_{n[il]} \mathbf{W}^{-1}_{n}\tilde{\mathbf{R}}_{il}^{1/2}{\mathbf{x}}_{inl}
\\=
&\frac{1}{N_{t}}\tilde{\mathbf{x}}_{inl}^H\tilde{\mathbf{R}}_{il}^{1/2}\mathbf{W}_{n[il]}^{-1}\mathbf{C}_{n[il]}\mathbf{W}^{-1}_{n}\tilde{\mathbf{R}}_{il}^{1/2}\tilde{\mathbf{x}}_{inl}
\\
+ &\frac{1}{N_{t}}{\mathbf{x}}_{inl}^H\tilde{\mathbf{R}}_{il}^
{1/2}(\mathbf{W}^{-1}_{n}-\mathbf{W}_{n[il]}^{-1})\mathbf{C}_{n[il]}\mathbf{W}^{-1}_{n}\tilde
{\mathbf{R}}_{il}^{1/2}{\mathbf{x}}_{inl}.
\end{aligned}
\end{equation}
Since $\mathbf{W}^{-1}_{n}-\mathbf{W}_{n[il]}^{-1}=-\mathbf{W}^{-1}_{n}(\mathbf{W}_{n}-\mathbf{W}_{n[il]})\mathbf{W}_{n[il]}^{-1}$ and $\mathbf{W}_{n}-\mathbf{W}_{n[il]}=\alpha_{inl}\tilde{\mathbf{R}}_{il}^{1/2}{\mathbf{x}}_{inl}{\mathbf{x}}_{inl}^H \tilde{\mathbf{R}}_{il}^{1/2}$, the above equation can be further reformulated as
\begin{equation}
\begin{aligned}
&\tilde{\mathbf{h}}_{inl}^{H} {\mathbf{W}}_{n}^{-1} \mathbf{C}_{n[il]}
{\mathbf{W}}_{n}^{-1}
\tilde{\mathbf{h}}_{inl}
\\
= & \frac{1}{N_{t}}{\mathbf{x}}_{inl}^H\mathbf{F}_{inl}{\mathbf{x}}_{inl} -  \frac{1}{N_{t}}\alpha_{inl}{\mathbf{x}}_{inl}^H
\mathbf{E}_{inl}{\mathbf{x}}_{inl}{\mathbf{x}}_{inl}^H\mathbf{F}_{inl}{\mathbf{x}}_{inl}
\end{aligned} \label{eq:de_interf_eq}
\end{equation}
where we have defined $\mathbf{E}_{inl}\triangleq \tilde{\mathbf{R}}_{il}^{1/2}\mathbf{W}^{-1}_{n}\tilde{\mathbf{R}}_{il}^{1/2}$ and $\mathbf{F}_{inl}\triangleq \tilde{\mathbf{R}}_{il}^
{1/2}\mathbf{W}_{n[il]}^{-1}{\mathbf{C}}_{n[il]}\mathbf{W}^{-1}_{n}\tilde{\mathbf{R}}_{il}^{1/2}$. Through lemma we obtain
\begin{align}
    {\mathbf{x}}_{inl}^H\mathbf{E}_{inl}{\mathbf{x}}_{inl}&-\frac{u_{inl}}{\alpha_{inl}u_{inl}+1}
    \mathop{\longrightarrow}_{N_t\to\infty} 0 
    \\
    {\mathbf{x}}_{inl}^H\mathbf{F}_{inl}{\mathbf{x}}_{inl}&-\frac{v_{inl}}{\alpha_{inl}u_{inl}+1}
    \mathop{\longrightarrow}_{N_t\to\infty} 0 
\end{align}
where we have denoted $u_{inl}=\frac{1}{N_{t}}\text{Tr}\left(\tilde{\mathbf{R}}_{il}\mathbf{W}_{n[il]}^{-1}\right)$ and $v_{inl}=\frac{1}{N_{t}}\text{Tr}\left( \tilde{\mathbf{R}}_{il}\mathbf{W}_{n[il]}^{-1}\mathbf{C}_{n[il]}\right)$. Plugging the above two equations into \eqref{eq:de_interf_eq} implies 
\begin{equation}
\tilde{\mathbf{h}}_{inl}^{H} {\mathbf{W}}_{n}^{-1} \mathbf{C}_{n[il]}
\tilde{\mathbf{W}}_{n}^{-1}
\tilde{\mathbf{h}}_{inl} - \frac{\alpha_{inl}uv-uv+v}{N_{t}(\alpha_{inl}u+1)^2} \mathop{\longrightarrow}_{N_t\to\infty} 0.
\end{equation} 
By further exploiting lemma \ref{lemma:rand_one_perturbation} to these two terms, we have 
\begin{align}
    u_{inl}-e_{inl}&
    \mathop{\longrightarrow}_{N_t\to\infty} 0
    \\
    \frac{1}{N_{t}}v_{inl}-z_{inl}&
    \mathop{\longrightarrow}_{N_t\to\infty} 0 
\end{align}
where $e_{inl}$ is the solution of \eqref{eq:e} and $z_{inl}$ is expressed by 
\begin{equation}
    z_{inl}=\frac{1}{N_{t}^2}\text{Tr}\left( \mathbf{B}_{n[il]}\tilde{\mathbf{H}}_{n[il]}\mathbf{W}_{n}^{-1}\tilde{\mathbf{R}}_{il}\mathbf{W}_{n}^{-1}\tilde{\mathbf{H}}_{n[il]}^H \right).
\end{equation}
Thus, it remains to derive the asymptotic approximation of $z_{inl}$. To achieve this, we first rewrite $z_{inl}$ as
\begin{equation}
    z_{inl}=\frac{1}{N_{t}}\sum_{(j, q) \neq(i, l)}
    \mathbf{B}_{n[jq]}
     {\mathbf{x}}_{jnq}^H \tilde{\mathbf{R}}_{jq}^{1/2}
     \mathbf{W}_{n}^{-1}\tilde{\mathbf{R}}_{il}\mathbf{W}_{n}^{-1}\tilde{\mathbf{R}}_{jq}^{1/2}{\mathbf{x}}_{jnq}
\end{equation}
By subsequently adopting lemma \ref{lemma:matrix_inversion} - \ref{lemma:rand_one_perturbation}, we have
\begin{equation}
    z_{inl}=\frac{1}{N_t}\sum_{(j, q) \neq(i, l)}
     \mathbf{B}_{n[jq]}\frac{\frac{1}{N_{t}}
     \text{Tr}\left(\tilde{\mathbf{R}}_{jq}\mathbf{W}_{n}^{-1}\tilde{\mathbf{R}}_{il}\mathbf{W}_{n}^{-1}\right)}{\left(1+\frac{1}{N_{t}}\alpha_{jnq}\text{Tr}\left(\tilde{\mathbf{R}}_{jq}\mathbf{C}_{n}^{-1}\right)\right)^2}.
\end{equation}
Note that the deterministic equivalent for the denominator of the right-hand side of the above equation can be directly derived similar to the steps in subsection \ref{subsec:norm_coff}. Therefore, we only focus on the numerator part in the following. By observing this part, we have
\begin{align}
    & \frac{1}{N_{t}}\text{Tr}\left(\tilde{\mathbf{R}}_{jq}\mathbf{W}_{n}^{-1}\tilde{\mathbf{R}}_{il}\mathbf{W}_{n}^{-1}\right) \\
     = & \frac{d}{dz}\frac{1}{N_{t}}\text{Tr}\left(\tilde{\mathbf{R}}_{jq}(\mathbf{\Gamma}_{n}+\beta \mathbf{I}_{N_{t}}-z\tilde{\mathbf{R}}_{il})^{-1}
     \right) |_{z=0} \text{.}
\end{align}
where $\boldsymbol{\Gamma}_{n} = (1/N_t) \cdot \sum_{j=1}^I \sum_{q=1}^L \alpha_{jnq} \tilde{\mathbf{h}}_{jnq}\tilde{\mathbf{h}}_{jnq}^H $.

Following Theorem \ref{theorem:de}, we have
\begin{equation}
\begin{aligned}
    & \frac{1}{N_t}\text{Tr}\left(\tilde{\mathbf{R}}_{jq}(\mathbf{\Gamma}_{n}+\beta \mathbf{I}_{N_{t}}-z\tilde{\mathbf{R}}_{il})^{-1}
     \right) \\  
    - &  \frac{1}{N_t} \text{Tr}\left(\tilde{\mathbf{R}}_{jq} \mathbf{T}_{inl}(z) \right) \mathop{\longrightarrow}_{N_t\to\infty} 0,
\end{aligned}
\end{equation}
where
\begin{equation}
    \mathbf{T}_{inl}(z) =\left( \frac{1}{N_t}\sum_{j=1}^I\sum_{q=1}^L \frac{\alpha_{jnq} \tilde{\mathbf{R}}_{jq}}{1+e_{jnq,inl}(z)}+\beta\mathbf{I} -z\tilde{\mathbf{R}}_{il} \right)^{-1} \label{eq:diff_T}
\end{equation}
with $e_{jnq,inl}(z)=(1/N_{t})\cdot \text{Tr} \left(\tilde{\mathbf{R}}_{jq}\mathbf{T}_{inl}(z)\right)$. Differentiating \eqref{eq:diff_T} w.r.t. $z$ implies
\begin{equation}
\begin{aligned}
    & \frac{d}{d z} \frac{1}{N_t}\text{Tr}\left(\tilde{\mathbf{R}}_{jq}(\mathbf{\Gamma}_{n}+\beta \mathbf{I}_{N_{t}}-z\tilde{\mathbf{R}}_{il})^{-1}
     \right) \\  
    - &  \frac{1}{N_t} \text{Tr}\left(\tilde{\mathbf{R}}_{jq} \mathbf{T}'_{inl}(z) \right) \mathop{\longrightarrow}_{N_t\to\infty} 0,
\end{aligned}
\end{equation}
where
\begin{equation}
\begin{aligned}
    & \mathbf{T}'_{inl}(z) = \frac{d}{d z} \mathbf{T}_{inl}(z) = \\
    & \mathbf{T}_{inl}(z)\left[\frac{1}{N_{t}}
    \sum_{j=1}^I\sum_{q=1}^L\frac{\alpha_{jnq}\tilde{\mathbf{R}}_{jq}e'_{jnq,inl}(z)}{(1+e_{jnq,inl}(z))^2}+\tilde{\mathbf{R}}_{il} \right]\mathbf{T}_{inl}(z) \text{.}
\end{aligned}
\end{equation}
Let $z = 0$, we get $e_{jnq}=e_{jnq,inl}(0)=\frac{1}{N_{t}}\text{Tr} \left( \tilde{\mathbf{R}}_{jq}\mathbf{T}_{n}\right)$, where $\mathbf{T}_{n}=\mathbf{T}_{inl}(0)$ as shown in \eqref{theorem:DE_SINR}, $e'_{jnq,inl}$ is the unique solution of the equations $e'_{jnq,inl} = (1/N_t)\cdot \text{Tr}( \tilde{\mathbf{R}}_{jq} \mathbf{T}'_{inl}(0) )$. Combining the above equations together, $z_{inl}$ can be explicitly expressed as
\begin{equation}
    z_{inl}=\sum_{(j, q) \neq(i, l)}
    \alpha_{jnq}p_{jnq}\frac{e'_{jnq,inl}}{e'_{jnq}}
    \text{.}
\end{equation}
Subsequently we can obtain
\begin{equation}
\begin{aligned}
    & \tilde{\mathbf{h}}_{inl}^{H} {\mathbf{W}}_{n}^{-1} \mathbf{C}_{n[il]}
    {\mathbf{W}}_{n}^{-1}
    \tilde{\mathbf{h}}_{inl} \\
    - & \frac{z_{inl}((\alpha_{inl}-1)e_{inl}+1)}{(1+\alpha_{inl}e_{inl})^2}
    \mathop{\longrightarrow}_{N_t\to\infty} 0.
\end{aligned}
\end{equation}

Putting the results in the aforementioned subsections together, we can derive the deterministic equivalent of $\gamma_{inl}$:
\begin{equation}
    \bar{\gamma}_{inl} = \frac{ N_{t}\alpha_{inl} p_{inl}
    e_{inl}^2(1+e_{inl})^2}{e'_{inl}
    \sum_{(j, q) \neq(i, l)}
    \frac{\alpha_{jnq}p_{jnq}e'_{jnq,inl}}{e'_{jnq}}+e'_{inl}\sigma_{in}^{2}(1+e_{inl})^2},
\end{equation}
which ends the proof.

\section{Proof of Theorem \ref{theorem:de_sinr_zf}} \label{app:theorem_de_sinr_zf}
Based on Assumption 3 and 4, we have
\begin{equation}
    \bar{\gamma}_{inl} 
    =\frac{ N_{t}\alpha_{inl}p_{inl}(\delta{e}_{inl})^2}
    {\frac{\delta^2 e'_{inl} \sum_{(j, q) \neq(i, l)}
    \frac{\alpha_{jnq}p_{jnq}e'_{jnq,inl}}{e'_{jnq}}}{(1+ e_{inl})^2}+\delta^2 e'_{inl}\sigma_{in}^{2}}
\end{equation}
Moreover, since $\text{lim}_{\delta\to 0}\delta^2e'_{inl}=\bar{e}_{inl}$, we can get
\begin{align}
\lim_{\delta \rightarrow 0} \bar{\gamma}_{inl} 
    =&\frac{N_{t}\alpha_{inl}p_{inl}(\bar{e}_{inl})^2}
    {\frac{\bar{e}_{inl} \sum_{(j, q) \neq(i, l)}
    \frac{\alpha_{jnq}p_{jnq}e'_{jnq,inl}}{e'_{jnq}}}{(1+ e_{inl})^2}+\bar{e}_{inl}\sigma_{in}^{2}}
    \\
 =&\frac{N_{t}\alpha_{inl}p_{inl}\bar{e}_{inl}}
    {\frac{\sum_{(j, q) \neq(i, l)}
    \alpha_{jnq}p_{jnq}\frac{e'_{jnq,inl}}{e'_{jnq}}}{(1+ e_{inl})^2}+\sigma_{in}^{2}} \\
= & \frac{N_{t}\alpha_{inl}p_{inl}\bar{e}_{inl}}{\sigma_{in}^{2}}\text{,} \label{DE_SINR2}
\end{align}
where
\begin{align}
    \bar{e}_{inl}=&\text{lim}_{\delta\to0}\frac{1}{N_{t}}\text{Tr}\left(\tilde{\mathbf{R}}_{il}\left(\frac{1}{N_t}\sum_{j,q}\frac{\alpha_{jnq}\tilde{\mathbf{R}}_{jq}}{\delta+\delta e_{jnq}}+\mathbf{I}_{N_t}\right)^{-1} \right)
    \\
    =&\frac{1}{N_{t}}\text{Tr}\tilde{\mathbf{R}}_{il}{\mathbf{T}}_n.
\end{align}

\ifCLASSOPTIONcaptionsoff
  \newpage
\fi



%
\bibliographystyle{ieeetr}
\bibliography{ref}

%




\end{document}